\documentclass[11pt]{article}
\usepackage[T1]{fontenc}
\usepackage[utf8]{inputenc}
\usepackage[a4paper,margin=2.5cm]{geometry}
\usepackage{amsmath,amssymb,amsthm,mathtools}
\usepackage[expansion=false]{microtype}
\usepackage{algorithm, algpseudocode, booktabs, tabularx, enumitem, xcolor, hyperref}
\hypersetup{colorlinks=true,linkcolor=blue!45!black,citecolor=blue!45!black,urlcolor=blue!45!black}

\theoremstyle{plain}
\newtheorem{theorem}{Theorem}
\newtheorem{lemma}[theorem]{Lemma}
\newtheorem{proposition}[theorem]{Proposition}
\newtheorem{corollary}[theorem]{Corollary}
\theoremstyle{definition}
\newtheorem{definition}{Definition}
\theoremstyle{remark}
\newtheorem{remark}{Remark}
\newcommand{\Rring}{\mathcal{R}}
\newcommand{\Rq}{\mathcal{R}_q}
\newcommand{\Adv}{\mathsf{Adv}}
\newcommand{\negl}{\mathsf{negl}}
\newcommand{\Btau}{B_\tau}
\newcommand{\us}{\mathrel{\overset{\$}{\leftarrow}}}
\newcommand{\inner}[1]{\langle #1 \rangle}
\newcommand{\Exp}{\mathsf{Exp}}
\newcommand{\Oracle}[1]{\mathcal{O}\mathsf{#1}}
\DeclareMathOperator{\HighBits}{\mathsf{HighBits}}
\DeclareMathOperator{\LowBits}{\mathsf{LowBits}}
\DeclareMathOperator{\Encode}{\mathsf{Encode}}
\DeclareMathOperator{\Setup}{\mathsf{Setup}}
\DeclareMathOperator{\KeyGen}{\mathsf{KeyGen}}
\DeclareMathOperator{\Sign}{\mathsf{Sign}}
\DeclareMathOperator{\Verify}{\mathsf{Verify}}
\DeclareMathOperator{\CSetup}{\mathsf{CSetup}}
\DeclareMathOperator{\Com}{\mathsf{Com}}
\DeclareMathOperator{\Shuffle}{\mathsf{Shuffle}}
\DeclareMathOperator{\WSign}{\mathsf{WSign}}
\DeclareMathOperator{\WSVerify}{\mathsf{WSVerify}}
\DeclareMathOperator{\Confirm}{\mathsf{Confirm}}
\DeclareMathOperator{\CVerify}{\mathsf{CVerify}}
\DeclareMathOperator{\Chk}{\mathsf{Check}}
\DeclareMathOperator{\AKeyGen}{\mathsf{AKeyGen}}
\DeclareMathOperator{\ASign}{\mathsf{ASign}}
\DeclareMathOperator{\AVerify}{\mathsf{AVerify}}
\DeclareMathOperator{\EWS}{\mathsf{EWS}}
\DeclareMathOperator{\DS}{\mathsf{DS}}
\DeclareMathOperator{\Ambi}{\mathsf{Ambi}}

\newcommand{\AdvEUF}[1]{\Adv^{\mathsf{EUF\text{-}CMA}}_{#1}}
\newcommand{\AdvMLWE}{\Adv^{\mathsf{MLWE}}}
\newcommand{\AdvMSIS}{\Adv^{\mathsf{MSIS}}}
\newcommand{\AdvSTMSIS}{\Adv^{\mathsf{STMSIS}}}
\newcommand{\Advhide}{\Adv^{\mathsf{hide}}}
\newcommand{\Advbind}{\Adv^{\mathsf{bind}}}
\newcommand{\Advanon}{\Adv^{\mathsf{anon}}}
\newcommand{\Advunf}{\Adv^{\mathsf{unf}}}
\newcommand{\Advwd}{\Adv^{\mathsf{wd}}}
\newcommand{\Advclaim}{\Adv^{\mathsf{claim}}}
\numberwithin{equation}{section}

\title{\bfseries Lattice-based extended withdrawability}
\author{Ramses Fernandez\\ Fairgate Labs\\ \texttt{ramses.fernandez@fairgate.io}}
\date{}
\begin{document}
\maketitle
\begin{abstract}
We extend the extended withdrawable signatures of Liu, Susilo and Baek to lattice-based constructions built on the Fiat-Shamir with aborts paradigm. Departing from an earlier draft that transported a per-signer shift in the clear, which leaks the signer, we realise extended withdrawable signatures as a claimable ring signature: signer ambiguity is provided by a one-out-of-$N$ signature used as a black box (anonymity under full key exposure), and confirmation is the signer's claim, a binding signature together with the opening of a hiding index commitment bound into the transcript. No signer-derived value is published in the clear. We give complete proofs of correctness, extended withdrawability (as anonymity-until-claim), unforgeability under insider corruption, and claimability soundness, reducing to decisional MLWE (commitment hiding), MSIS (commitment binding), the anonymity of the one-out-of-$N$ scheme, and the EUF-CMA security of the base signature, in the (quantum) random-oracle model. We instantiate the base signature with a no-hint, full-$t$ Dilithium-style scheme and the one-out-of-$N$ layer with an established lattice one-out-of-many proof.
\end{abstract}

\section{Introduction}\label{sec:intro}
Digital signatures bind a signer's identity to a message and provide authenticity, integrity and non-repudiation, underpinning secure communication and, increasingly, the settlement of transactions in decentralised systems such as blockchains. A defining feature of a conventional signature is its permanence. Once a signer legitimately produces a signature on a message, that signature remains verifiable under the signer's public key indefinitely, and the signer cannot rescind it. Permanence is usually a feature, but it is a liability in the many settings where a commitment may legitimately need to be revoked: limited-time access grants, revocable electronic agreements, service-level commitments in outsourcing systems, decentralised escrow, and privacy-preserving voting. In all of these one would like a signer to be able to withdraw a signature without disclosing its secret key and without affecting the validity of its other signatures.

Liu, Baek and Susilo introduced withdrawable signatures \cite{liu2023} to capture exactly this capability. A withdrawable signature is revoked by default: after issuing it the signer may simply take no further action, in which case the signature is not verifiable under the signer's public key. If the signer later wishes to stand behind the message, it, and only it, can confirm the withdrawable signature, turning it into a publicly verifiable one that remains deterministically traceable to the original object, so that the confirmed signature cannot be detached from the withdrawable one it came from. The first construction realised withdrawability through a designated-verifier signature (DVS) \cite{jakobsson1996,chaum1991}: the withdrawable signature is verifiable only by a single, signer-chosen verifier, and the ambiguity between the signer and that designated verifier is precisely what allows the signature to be withdrawn. The price is a significant limitation, verification is restricted to one designated party, so the scheme lacks universal verifiability.

\paragraph{From designated verifiers to extended withdrawability.}
To lift this restriction, Liu, Susilo and Baek proposed extended withdrawable signatures \cite{liu2024}. Their starting observation is that a DVS is a two-party object whose ambiguity is exactly that of a $1$-out-of-$2$ signature, and that the notion generalises to $N\ge 2$ potential signers: a withdrawable signature can be made verifiable under any public key in a set of potential signers $\pi$ while keeping the true signer ambiguous among them, until confirmation. This ``extended withdrawability'' can be obtained from any $1$-out-of-$N$ signature that provides signer ambiguity, ring signatures \cite{rivest2001} and designated-verifier signatures being two natural instances, and yields a publicly verifiable withdrawable signature: any party holding the signature and the public-key set can check it, yet no party can tell which member of $\pi$ signed. Liu, Susilo and Baek give a generic construction of extended withdrawable signatures from discrete-logarithm primitives, with an instantiation based on the Schnorr signature \cite{schnorr1991}, and prove it correct, unforgeable under insider corruption, and extended-withdrawable.

\paragraph{The post-quantum gap.}
The discrete-logarithm foundation of \cite{liu2024}, and of the Schnorr and BLS \cite{boneh2001} building blocks it relies on, is broken by Shor's quantum algorithm \cite{shor1999}. Because withdrawable signatures are meant to protect long-lived commitments in decentralised systems, a construction that fails against a quantum adversary is of limited durable value, and a post-quantum realisation is called for. The natural route is Lyubashevsky's Fiat-Shamir with aborts paradigm \cite{lyubashevsky2009}, which turns lattice-based identification into signatures through a controlled rejection-sampling (``abort'') step and underlies the standardised signature Dilithium (ML-DSA) \cite{ducas2017} as well as HAETAE \cite{cheon2024}. A first step in this direction handled the designated-verifier (two-party) withdrawable signature over the Fiat-Shamir with aborts paradigm \cite{fernandez2026}; the present work addresses the extended, publicly verifiable, $N$-signer notion.

\paragraph{Why a literal lattice transcription fails.}
Porting \cite{liu2024} to lattices is not just a matter of replacing Schnorr by a lattice signature. The discrete-logarithm construction re-randomises each ring member's public key by a multiplicative shift and publishes those shifts, so that the confirmed signature can later be linked back to the withdrawable one. Over a prime-order group this is safe, because the published shifts $\{g^{s_j}\}$ are independent of the signer's secret key and uniformly distributed, and therefore leak nothing about who signed. A literal lattice transcription, however, replaces each shift $g^{s_j}$ by a term of the form $A\mathbf s_1 s_j$ that carries the signer's own secret $\mathbf s_1$ as a common factor. Since the signer's public key satisfies $t_i=A\mathbf s_1+\mathbf s_2$ with $\mathbf s_2$ short, the value $A\mathbf s_1$ is within short distance of $t_i$ but of no other public key; hence publishing $\{A\mathbf s_1 s_j\}$ exposes $A\mathbf s_1\approx t_i$ and breaks the very signer ambiguity the scheme exists to provide. This is the same structural obstruction, the lack of a free re-randomisation of a shared value in the lattice setting, encountered for the designated-verifier case \cite{fernandez2026}, and it means the transformation must be redesigned rather than transcribed.

\paragraph{Our approach.}
We keep the two objectives of an extended withdrawable signature separate and realise each by a primitive with an established security analysis, so that no signer-derived value is ever published in the clear. Ambiguity is provided by an anonymous $1$-out-of-$N$ (ring) signature, used as a black box and instantiated by an established lattice one-out-of-many proof \cite{groth2015,esgin2019}. Confirmation is the signer's claim: a binding signature on the withdrawable object, together with the opening of a hiding commitment to the signer's index that is bound into the transcript at signing time. Before confirmation the commitment hides the signer and the ring signature is anonymous; at confirmation the signer opens the commitment, and the binding signature ties the claim to the exact withdrawable object. This is precisely a claimable ring signature in the sense of Park and Sealfon \cite{park2019}, a primitive that Liu, Susilo and Baek themselves single out as an instantiation of withdrawable signatures \cite{liu2024,yamashita2023}. The resulting scheme is publicly verifiable, ambiguous until claimed, and provably secure from decisional MLWE, MSIS, and the unforgeability of the base signature.

\subsection{Related work}\label{sec:related}
\paragraph{Withdrawable signatures.}
Withdrawable signatures were introduced in \cite{liu2023} with a designated-verifier construction, and extended to the publicly verifiable, $N$-signer setting in \cite{liu2024}. A lattice-based designated-verifier variant over the Fiat-Shamir with aborts paradigm was given in \cite{fernandez2026}; the present paper is its extended, publicly verifiable counterpart.

\paragraph{Designated-verifier and universal designated-verifier signatures.}
Designated-verifier signatures were introduced independently by Jakobsson, Sako and Impagliazzo \cite{jakobsson1996} and by Chaum \cite{chaum1991}, and studied extensively thereafter \cite{li2007,thorn2020}. Universal designated-verifier signatures (UDVS), in which any holder of a signature can designate a verifier, were initiated by Steinfeld et al. \cite{steinfeld2003} and extended to Schnorr and RSA signatures in \cite{steinfeld2004}; a lattice-based UDVS appears in \cite{li2018}.

\paragraph{Ring signatures and one-out-of-many proofs.}
Ring signatures, introduced by Rivest, Shamir and Tauman \cite{rivest2001}, let a signer produce a signature on behalf of a self-chosen set of public keys while hiding which member signed; constructions exist from RSA, discrete-logarithm \cite{herranz2003}, pairing \cite{boneh2003} and lattice \cite{esgin2019} assumptions. A generic route is the one-out-of-many zero-knowledge proof of Groth and Kohlweiss \cite{groth2015}, with a short lattice instantiation by Esgin et al. \cite{esgin2019}. Our $1$-out-of-$N$ building block is any such scheme with anonymity under full key exposure.

\paragraph{Claimable ring and designated-verifier signatures.}
Park and Sealfon \cite{park2019} introduced claimable ring signatures, which let an anonymous signer later prove authorship, and Yamashita et al. \cite{yamashita2023} the analogous claimable designated-verifier signatures. As Liu, Susilo and Baek observe \cite{liu2024}, both are instantiations of withdrawable signatures: the ``claim'' is exactly the confirmation step. Our construction makes this correspondence its design principle.

\paragraph{Lattice signatures from Fiat-Shamir with aborts.}
The Fiat-Shamir heuristic \cite{fiat2000} and its lattice adaptation with aborts \cite{lyubashevsky2009} underlie Dilithium \cite{ducas2017} and HAETAE \cite{cheon2024}, whose security rests on module lattice problems and, in the quantum random-oracle model, on the analyses of Kiltz, Lyubashevsky and Schaffner \cite{kiltz2018} and the adaptive-reprogramming lemma of Grilo, H\"ovelmanns, H\"ulsing and Majenz \cite{grilo2021}. Our commitment layer uses an Ajtai-style commitment \cite{ajtai1996}.

\subsection{Contributions}\label{sec:contributions}
Concretely, this paper makes the following contributions.

\begin{enumerate}
    \item Security model: We adopt the syntax and the three security goals, correctness, unforgeability under insider corruption, and extended withdrawability, of Liu, Susilo and Baek \cite{liu2024}, and we sharpen their formalisation of extended withdrawability. Their definition asserts only that a withdrawable signature can be re-attributed to another member and still verify (a completeness statement), delegating the hiding guarantee to the signer ambiguity of the underlying $1$-out-of-$N$ scheme. We fold both into a single anonymity-until-claim indistinguishability game (Definition \ref{def:wd}), state unforgeability under insider corruption as an explicit experiment (Definition \ref{def:eufcmaEWS}), and add a claimability soundness notion (Definition \ref{def:claim}) that makes precise the requirement that only the true author can confirm.
    \item Generic construction: We give a generic construction of extended withdrawable signatures over the Fiat-Shamir with aborts paradigm (Section \ref{sec:construction}) from three modular ingredients: any base signature, any anonymous $1$-out-of-$N$ signature, and any hiding-and-binding commitment. Crucially, and in contrast with a literal transcription of the discrete-logarithm scheme, the construction publishes no signer-derived value: ambiguity comes from the ring signature and traceability from a committed index, avoiding the signer-identifying leak described above.
    \item Complete security proofs: We prove, with fully explicit game-based reductions in the (quantum) random-oracle model (Section \ref{sec:security}): correctness; extended withdrawability, reducing to decisional MLWE (commitment hiding) and the anonymity of the $1$-out-of-$N$ scheme; unforgeability under insider corruption, reducing to the EUF-CMA security of the base signature (with no rewinding and no multiplicative random-oracle loss); and claimability soundness, reducing to MSIS (commitment binding).
    \item Instantiation: We instantiate the base signature with a no-hint, full-$t$ Dilithium-style scheme and the ambiguity layer with a lattice one-out-of-many proof (Section \ref{sec:instantiation}), correcting the parameter choices (the low-order bound $\gamma_2$, the challenge set $\Btau$, and the high/low-order decomposition) and explaining why a naive Dilithium challenge-split is not a sound anonymous ring signature. HAETAE is discussed as an alternative base.
    \item Comparison with Liu-Susilo-Baek: We give a precise comparison (Section \ref{sec:comparison}) that traces every divergence from \cite{liu2024} to a single structural fact, the absence of a free re-randomisation of a shared value in the lattice setting, and shows that our claimable-ring-signature design is faithful to their primitive, their security goals, and their own stated intuition.
\end{enumerate}

\section{Preliminaries}\label{sec:prelim}
\subsection{Notation, rings and norms}\label{sec:notation}
We write $\kappa$ for the security parameter. A function $\negl:\mathbb N\to\mathbb R_{\ge 0}$ is negligible if it decays faster than any inverse polynomial; PPT abbreviates probabilistic polynomial-time. For a positive integer $a$ we write $[a]=\{1,\dots,a\}$; in particular $[N]$ indexes the $N$ potential signers of a ring, while $n$ (never bracketed) always denotes the ring degree of $\Rq$. Sets of public keys carry a fixed canonical total order (say, lexicographic on the encoding $\inner{\mathbf t}$); the \emph{index} of a key in a ring $\pi$ always means its position under this order, and when keys are generated in an experiment we relabel them so that $pk_1<\cdots<pk_N$, making global labels coincide with positions in the full pool. We write $a\us S$ for a uniform draw of $a$ from a finite set $S$. All logarithms are base $2$, and $\inner{\cdot}$ denotes a fixed injective encoding of tuples of bit strings and ring elements into $\{0,1\}^\ast$.

Let $n$ be a power of $2$ and $q$ an odd prime (with $q\equiv 1\bmod 2n$ in the Dilithium instantiation, so that $x^n+1$ splits completely). We work in the rings $\Rring=\mathbb Z[x]/(x^n+1)$ and $\Rq=\mathbb Z_q[x]/(x^n+1)$. Elements of $\Rring,\Rq$ are written in regular font; vectors are bold lower-case and are column vectors unless stated otherwise; matrices are bold upper-case; $\mathbf I$ is the identity.

For an even positive integer $\alpha$ we write $r\bmod^{\pm}\alpha$ for the unique representative of $r\bmod\alpha$ in $(-\tfrac{\alpha}{2},\tfrac{\alpha}{2}]$, and for odd $\alpha$ in $[-\tfrac{\alpha-1}{2},\tfrac{\alpha-1}{2}]$ (centred reduction); $r\bmod^{+}\alpha$ is the representative in $[0,\alpha)$. For $w\in\mathbb Z_q$ set $\lVert w\rVert_\infty=\lvert w\bmod^{\pm}q\rvert$. For $w=\sum_{i=0}^{n-1}w_i x^i\in\Rring$ set $\lVert w\rVert_\infty=\max_i\lVert w_i\rVert_\infty$ and $\lVert w\rVert=(\sum_i\lVert w_i\rVert_\infty^2)^{1/2}$, and for $\mathbf w=(w_1,\dots,w_k)\in\Rring^k$, $\lVert\mathbf w\rVert_\infty=\max_i\lVert w_i\rVert_\infty$ and $\lVert\mathbf w\rVert=(\sum_i\lVert w_i\rVert^2)^{1/2}$. For $\eta\in\mathbb N$ let $S_\eta=\{w\in\Rring:\lVert w\rVert_\infty\le\eta\}$; when the exact bound is dictated by the underlying scheme we write simply $S$. The challenge set $\Btau\subset\Rring$ is the set of elements with exactly $\tau$ coefficients in $\{-1,1\}$ and the rest $0$; then $\lvert\Btau\rvert=2^{\tau}\binom{n}{\tau}$ and every $c\in\Btau$ has $\lVert c\rVert_\infty=1$ and $\lVert c\mathbf s\rVert_\infty\le\tau\lVert\mathbf s\rVert_\infty$.

\subsection{The (quantum) random-oracle model}\label{sec:qrom}
A hash function $H$ used inside a Fiat-Shamir transform is modelled as a random oracle: a function
drawn uniformly from all functions of the given domain and range, accessible to all parties only as an
oracle. In the classical random-oracle model (ROM) queries are classical; in the quantum random-oracle
model (QROM) an adversary may query $H$ on superpositions of inputs. Security reductions may program
the oracle, answering a chosen input with a chosen (consistently distributed) value; in the QROM this
is governed by the adaptive-reprogramming lemma of \cite{grilo2021}. Unless noted, all results are stated
in the classical ROM and lifted to the QROM in Remark \ref{rem:qrom}.

\subsection{High-order and low-order bits}\label{sec:bits}
Following Dilithium \cite{ducas2017}, for an even $\alpha\mid(q-1)$ the decomposition of $w\in\mathbb Z_q$ writes $w=\HighBits(w,\alpha)\cdot\alpha+\LowBits(w,\alpha)$ with $\LowBits(w,\alpha)=w\bmod^{\pm}\alpha$ (and a boundary correction ensuring $\HighBits\in\{0,\dots,(q-1)/\alpha\}$); the maps extend coordinatewise to $\Rring$ and to vectors. Intuitively $\HighBits$ keeps the top part of $w$ and $\LowBits$ the signed remainder of magnitude $\le\alpha/2$. We use the following standard stability property, which is the basis of correctness of every Fiat-Shamir with aborts scheme.

\begin{lemma}[Rounding stability \cite{ducas2017}]\label{lem:rounding}
Let $\alpha$ be an even divisor of $q-1$. For all $\mathbf w\in\Rq^k$ and $\mathbf u\in\Rq^k$ with $\lVert\mathbf u\rVert_\infty\le\beta$, if $\lVert\LowBits(\mathbf w,\alpha)\rVert_\infty<\tfrac{\alpha}{2}-\beta$ then $\HighBits(\mathbf w+\mathbf u,\alpha)=\HighBits(\mathbf w,\alpha)$.
\end{lemma}

\subsection{Digital signatures and EUF-CMA}\label{sec:sig}
\begin{definition}[Signature scheme]\label{def:sig}
A digital signature scheme is a triple $\DS=(\KeyGen,\Sign,\Verify)$: $\KeyGen(pp)$ outputs $(pk,sk)$; $\Sign(\mu,sk)$ outputs a signature $\delta$; $\Verify(\mu,pk,\delta)\in\{0,1\}$. It is correct if $\Verify(\mu,pk,\Sign(\mu,sk))=1$ with overwhelming probability for honestly generated keys and every $\mu$.
\end{definition}

\begin{definition}[EUF-CMA]\label{def:eufcma}
In $\Exp^{\mathsf{EUF\text{-}CMA}}_{\DS,\mathcal A}$, the challenger runs $(pk,sk)\gets\KeyGen(pp)$, gives $pk$ to $\mathcal A$, and answers signing queries $\mu$ by $\Sign(\mu,sk)$; let $Q$ be the set of queried messages. $\mathcal A$ wins by outputting $(\mu^\ast,\delta^\ast)$ with $\mu^\ast\notin Q$ and $\Verify(\mu^\ast,pk,\delta^\ast)=1$. We write $\AdvEUF{\DS}(\mathcal A)=\Pr[\mathcal A\text{ wins}]$, and say $\DS$ is EUF-CMA secure if this is negligible for every PPT $\mathcal A$.
\end{definition}

\subsection{Fiat-Shamir with aborts signatures}\label{sec:fswa}
Our base signatures follow Lyubashevsky's Fiat-Shamir with aborts paradigm \cite{lyubashevsky2009}. Keys are $pk=\mathbf t=A\mathbf s_1+\mathbf s_2$ with short $(\mathbf s_1,\mathbf s_2)$ and a public matrix $A$ (Section \ref{sec:pp}); signing draws a masking vector $\mathbf y\us S^l_{\gamma_1-1}$, forms the commitment $\mathbf w=\HighBits(A\mathbf y,2\gamma_2)$, derives a challenge $c=H(\mu,\mathbf w)\in\Btau$, sets the response $\mathbf z=\mathbf y+c\mathbf s_1$, and rejects (restarting) unless $\mathbf z$ and $\LowBits(A\mathbf y-c\mathbf s_2,2\gamma_2)$ lie within scheme-dependent bounds; verifying recomputes $\HighBits(A\mathbf z-c\mathbf t,2\gamma_2)$ and checks it hashes to $c$. The rejection step makes the output distribution independent of the secret key: the accepted transcript admits an efficient, witness-free simulator (non-abort honest-verifier zero-knowledge, naHVZK), a property we invoke for the ambiguity layer. Dilithium \cite{ducas2017} and HAETAE \cite{cheon2024} are concrete instances; we use such a scheme as the black-box base $\DS$ of Definition \ref{def:sig}.

\subsection{Hardness assumptions}\label{sec:hardness}
Security rests on three module-lattice assumptions. Informally, MLWE guards secret keys against key recovery, SelfTargetMSIS is the assumption behind new-message unforgeability of Fiat-Shamir with aborts signatures, and MSIS underlies strong unforgeability (resistance to malleability of existing signatures).

\begin{definition}[Module Learning With Errors, $\mathsf{MLWE}_{m,k,D}$]\label{def:mlwe}
For a distribution $D$ over $\Rq$, the advantage of $\mathcal A$ in distinguishing a uniform pair from a module-LWE pair is
\[
\AdvMLWE_{m,k,D}(\mathcal A)=\bigl\lvert\Pr[\mathcal A(A,\mathbf t)=1]-\Pr[\mathcal A(A,A\mathbf s_1+\mathbf s_2)=1]\bigr\rvert,
\]
over $A\us\Rq^{m\times k}$, $\mathbf t\us\Rq^m$, $\mathbf s_1\us D^k$, $\mathbf s_2\us D^m$. MLWE is hard if this is negligible for every PPT $\mathcal A$.
\end{definition}

\begin{definition}[Module Short Integer Solution, $\mathsf{MSIS}_{m,k,\gamma}$]\label{def:msis}
In Hermite normal form,
\[
\AdvMSIS_{m,k,\gamma}(\mathcal A)=\Pr\bigl[\,0<\lVert\mathbf y\rVert_\infty\le\gamma\ \wedge\ [\,\mathbf I\mid A\,]\mathbf y=\mathbf 0:A\us\Rq^{m\times k},\ \mathbf y\gets\mathcal A(A)\,\bigr],
\]
the probability of finding a short nonzero vector in the kernel of $[\,\mathbf I\mid A\,]$. MSIS is hard if this is negligible for every PPT $\mathcal A$.
\end{definition}

\begin{definition}[SelfTargetMSIS]\label{def:stmsis}
For a hash function $H:\{0,1\}^\ast\to\Btau$ modelled as a random oracle,
\[
\AdvSTMSIS_{H,m,k,\gamma}(\mathcal A)=\Pr\bigl[\,0<\lVert\mathbf y\rVert_\infty\le\gamma\ \wedge\ H(\inner{[\,\mathbf I\mid A\,]\mathbf y,\,M})=c:A\us\Rq^{m\times k},\ (\mathbf y=(\mathbf r,c),M)\gets\mathcal A^H(A)\,\bigr].
\]
SelfTargetMSIS is hard if this is negligible for every PPT $\mathcal A$ with polynomially many oracle queries.
\end{definition}

All three are assumed hard for the parameter sets of Section \ref{sec:instantiation}.

\subsection{Commitment schemes}\label{sec:commit}
\begin{definition}[Commitment]\label{def:commit}
A commitment scheme is a pair $(\CSetup,\Com)$: $\CSetup(1^\kappa)$ outputs a key $ck$, and $\Com_{ck}(x;\rho)$ maps a message $x$ and randomness $\rho$ to a commitment. An opening of a commitment $C$ is a pair $(x,\rho)$ with $C=\Com_{ck}(x;\rho)$.
\end{definition}

\begin{definition}[Hiding]\label{def:hiding}
$\Advhide_{\Com}(\mathcal A)=\bigl\lvert\Pr[\mathcal A(ck,\Com_{ck}(x_0;\rho))=1]-\Pr[\mathcal A(ck,\Com_{ck}(x_1;\rho))=1]\bigr\rvert$ maximised over $\mathcal A$-chosen $x_0,x_1$, with $ck\gets\CSetup(1^\kappa)$ and fresh $\rho$. The scheme is hiding if this is negligible for every PPT $\mathcal A$.
\end{definition}

\begin{definition}[Binding]\label{def:binding}
$\Advbind_{\Com}(\mathcal A)=\Pr[\,C=\Com_{ck}(x;\rho)=\Com_{ck}(x',\rho')\ \wedge\ x\ne x'\,]$ over $ck\gets\CSetup(1^\kappa)$ and $(C,x,\rho,x',\rho')\gets\mathcal A(ck)$. The scheme is binding if this is negligible for every PPT $\mathcal A$.
\end{definition}

\begin{remark}[Ajtai instantiation]\label{rem:ajtai}
We use the Ajtai commitment \cite{ajtai1996} $\Com_{ck}(x;\rho)=\bar A\rho+\Encode(x)$ with $ck=\bar A\us\Rq^{k'\times m}$ and short $\rho\in S^m$, where $\Encode$ is an injective encoding of the index into a fixed set of \emph{short} coset representatives, e.g.\ the $\lceil\log_2 N\rceil$-bit expansion of the index as a $0/1$-coefficient polynomial, so that $\lVert\Encode(i)\rVert_\infty\le B_{\Encode}$ for a small $B_{\Encode}$ and hence $\lVert\Encode(i)-\Encode(i')\rVert_\infty\le 2B_{\Encode}$. It is hiding under decisional MLWE: writing $\bar A=[\,\bar A_1\mid\bar A_2\,]$ with $\bar A_1$ invertible (which holds with overwhelming probability), $\bar A_1^{-1}\bar A\rho=\rho_1+(\bar A_1^{-1}\bar A_2)\rho_2$ is an MLWE sample in secret $\rho_2$, so $\bar A\rho$ is pseudorandom and masks the message, giving $\Advhide_{\Com}\le\AdvMLWE$. It is binding under MSIS: two openings of one commitment give $\bar A(\rho-\rho')=\Encode(x')-\Encode(x)$, i.e.\ a short nonzero solution to $[\,\mathbf I\mid\bar A\,]$, so $\Advbind_{\Com}\le\AdvMSIS$.
\end{remark}

\subsection{Public parameters}\label{sec:pp}
Throughout, $pp\gets\Setup(1^\kappa)$ fixes a common matrix $A\in\Rq^{k\times l}$ and a commitment key $ck$; every user key is generated relative to the same $A$. This shared $A$ is essential: the one-out-of-$N$ proof is a membership proof for the relation $\{\,\mathbf t:\exists\,(\mathbf s_1,\mathbf s_2)\ \text{short},\ \mathbf t=A\mathbf s_1+\mathbf s_2\,\}$, which is a single well-formed relation across the ring only when every $\mathbf t_j$ is generated against the same $A$.
\section{Building block: an anonymous one-out-of-\texorpdfstring{$N$}{N} signature}\label{sec:buildingblock}
A one-out-of-$N$ signature lets a signer produce, for a self-chosen set $\pi=\{\mathbf t_1,\dots,\mathbf t_N\}$ of public keys containing its own, a signature that verifies under $\pi$ yet hides which member signed \cite{liu2024}. It is the engine that supplies ambiguity to an extended withdrawable signature.

\subsection{Role and design decision}\label{sec:role}
Extended withdrawability is delivered by an anonymous one-out-of-$N$ signature: the withdrawable signature is such a signature over the set of potential signers, and being withdrawable is exactly being signer-ambiguous. This is the design of Liu, Susilo and Baek \cite{liu2024}, who realise the one-out-of-$N$ signature by a concrete construction from discrete-logarithm primitives (Algorithm \ref{alg:dl}) and prove its signer ambiguity directly. We depart from their treatment in one respect: we specify the one-out-of-$N$ signature abstractly, through the two security properties we actually use (Definitions \ref{def:oneofn}--\ref{def:ambiunf}), and instantiate it with an established lattice one-out-of-many proof. There are two reasons, one methodological and one forced by the setting.
\begin{itemize}
    \item Modularity: Fixing the exact interface separates the security of the withdrawable layer from that of the ambiguity engine: the reductions of Section \ref{sec:security} then rest on a single, clearly stated assumption, signer ambiguity under full key exposure, rather than on the internals of one particular ring signature. Any scheme meeting the interface may be plugged in.
    \item Necessity: The specific discrete-logarithm construction of \cite{liu2024} has no secure literal lattice analogue (Proposition \ref{prop:notranscription} below). A faithful post-quantum port therefore cannot transcribe it; it must use a genuinely lattice-based ring-signature technique, which is most naturally invoked as a black box. Treating $\Ambi$ abstractly is thus not only cleaner but unavoidable.
\end{itemize}

\subsection{The discrete-logarithm one-out-of-\texorpdfstring{$N$}{N} signature of Liu-Susilo-Baek}\label{sec:dl}
For reference we recall the construction of \cite{liu2024}. Let $G$ be a group of prime order $p$ with generator $g$, $H:\{0,1\}^\ast\to\mathbb Z_p$, and $pk_j=g^{sk_j}$.

\begin{algorithm}[h]
\caption{Discrete-logarithm one-out-of-$N$ signature \cite{liu2024}}\label{alg:dl}
\begin{algorithmic}[1]
\Procedure{$\Ambi.\Sign$}{$m,sk_i,\gamma=\{pk_j\}_{j=1}^N$}
    \State $e\us\mathbb Z_p^\ast$;\quad for $j\ne i$: $t_j\us\mathbb Z_p^\ast$
    \State $U=g^e\prod_{j\ne i}pk_j^{t_j}$;\quad $t=H(m,\gamma,U)$
    \State $t_i=t-\sum_{j\ne i}t_j$;\quad $z=e-sk_i\,t_i$
    \State \Return $\sigma=(t_1,\dots,t_N,z)$
\EndProcedure
\Procedure{$\Ambi.\Verify$}{$m,\gamma,\sigma$}
    \State parse $(t_1,\dots,t_N,z)$;\quad $t=\sum_j t_j$;\quad $U'=g^z\prod_j pk_j^{t_j}$
    \State \Return $[\,t=H(m,\gamma,U')\,]$
\EndProcedure
\end{algorithmic}
\end{algorithm}

Correctness is immediate: $U'=g^z\prod_j pk_j^{t_j}=g^{e-sk_i t_i}\prod_j g^{sk_j t_j}=g^e\prod_{j\ne i}pk_j^{t_j}=U$, and $\sum_j t_j=t$. Signer ambiguity holds because the challenge space is a field: a witness-free simulator draws $t_1,\dots,t_N,z\us\mathbb Z_p$, sets $U'=g^z\prod_j pk_j^{t_j}$, and programs $H(m,\gamma,U'):=\sum_j t_j$; its output is distributed exactly as a real signature and is independent of the signer index. Two features make this work, and both are lost over lattices: the recombined challenge $t_i=t-\sum_{j\ne i}t_j$ may take any value in $\mathbb Z_p$, and the response $z=e-sk_i t_i$ is a uniform field element (as $e$ is uniform), so no rejection is needed.

\subsection{Why a literal lattice transcription fails}\label{sec:whyfails}
The Fiat-Shamir with aborts setting keeps challenges in the sparse set $\Btau$ and responses short (this is what makes the response leak no information and the proof sound). These requirements are incompatible with the additive challenge-splitting of Algorithm \ref{alg:dl}.

\begin{proposition}[No secure literal lattice transcription]\label{prop:notranscription}
Consider the module-lattice transcription of Algorithm \ref{alg:dl} in which the group map $g^{(\cdot)}$ is replaced by $s\mapsto A s$, exponents by ring elements, challenges are required to lie in $\Btau$, and the signer recombines $c_i=c-\sum_{j\ne i}c_j$ with $c=H(\cdots)\in\Btau$, accepting only if $c_i\in\Btau$. Then for a ring of $N$ potential signers with $N\ge 4$, an honest execution is accepted with probability at most $2^{-\Omega(n)}$ in the constant-relative-weight regime $\tau=\Theta(n)$, where $n$ is the ring degree of $\Rq$; signing is therefore infeasible. (For $\tau$ fixed independent of $n$ the unconditional obstruction is instead the anonymity failure of Remark \ref{rem:anonfail}.)
\end{proposition}

\begin{proof}
Acceptance requires $c_i=c-\mathbf s\in\Btau$ for $\mathbf s=\sum_{j\ne i}c_j$. In the random-oracle model the challenge $c=H(\cdots)\in\Btau$ is a fresh uniform element drawn \emph{after} the summands $c_j$ (hence after $\mathbf s$) are fixed, so we bound $\Pr_{c\us\Btau}[\,c-\mathbf s\in\Btau\mid\mathbf s\,]$. Membership forces $(c-\mathbf s)_k\in\{-1,0,1\}$ at every coordinate $k$. A uniform $\Btau$ element is nonzero at $k$ with probability $p:=\tau/n$ and, given nonzero, is $\pm1$ with probability $\tfrac12$ each; thus each of the $m:=N-1\ge3$ summands has $k$-th coordinate $0$ with probability $1-p$ and $\pm1$ with probability $p/2$ each. Consider a coordinate with $|s_k|\ge2$: if $|s_k|\ge3$ then $|(c-\mathbf s)_k|\ge|s_k|-\lVert c\rVert_\infty\ge2$ and acceptance is impossible; if $|s_k|=2$ then $(c-\mathbf s)_k\in\{-1,0,1\}$ forces $c_k=\operatorname{sign}(s_k)$, an event of probability $p/2$ over $c$. Writing $A=\{k:|s_k|=2\}$, acceptance requires $c$ to meet $|A|$ support-and-sign constraints; since the support indicators of a fixed-weight $c$ are negatively associated and its signs are independent, $\Pr[\text{accept}\mid\mathbf s]\le(p/2)^{|A|}$. For $m=3$ one has $\Pr[|s_k|=2]=\Theta(p^2)$, so $\mathbb E|A|=\Theta(p^2 n)$. In the constant-relative-weight regime $\tau=\Theta(n)$ (thus $p=\Theta(1)$, $p/2<\tfrac12$) this gives $\mathbb E|A|=\Theta(n)$. The coordinates of a fixed-weight vector are not independent, so we conclude by Azuma--Hoeffding over the $3\tau$ sequential support/sign choices of the three summands (each choice moves $|A|$ by at most $2$), giving $\Pr[\,|A|<c\,n\,]\le 2^{-\Omega(n)}$ for a suitable constant $c>0$; splitting on this event, $\Pr[\text{accept}]\le 2^{-\Omega(n)}+(p/2)^{c\,n}=2^{-\Omega(n)}$ and signing is infeasible.
\end{proof}

\begin{remark}[Anonymity fails even if completeness is patched]\label{rem:anonfail}
Suppose one nonetheless forced acceptance (e.g.\ by enlarging the challenge set). The witness-free simulator that establishes anonymity of Algorithm \ref{alg:dl} still has no lattice analogue: it relies on $t_i$ ranging over the entire challenge space and on $z$ being a uniform field element, whereas here $c_i$ must be sparse and $z$ short and passes a witness-dependent low-order rejection test. The recombined $c_i$ is then not distributed as a fresh $\Btau$ challenge, and the simulated and real signer branches are distinguishable. This is the reason we do not transcribe Algorithm \ref{alg:dl} but instead invoke a one-out-of-many proof, whose anonymity is a bona fide zero-knowledge statement.
\end{remark}

\subsection{Abstract syntax and security}\label{sec:abstract}
We now record the interface used in the sequel.

\begin{definition}[One-out-of-$N$ signature]\label{def:oneofn}
A one-out-of-$N$ signature scheme over the common $A$ is a triple of algorithms $\Ambi=(\AKeyGen,\ASign,\AVerify)$. Here $\AKeyGen(pp)$ outputs a key pair, we take $\AKeyGen=\KeyGen$ of Section \ref{sec:construction}, so $\mathbf t=A\mathbf s_1+\mathbf s_2$; $\ASign(\mu,sk_i,\pi)$, with $i$ the signer's index in $\pi$, outputs a signature $\sigma$; and $\AVerify(\mu,\pi,\sigma)\in\{0,1\}$. It is correct if, whenever $sk_i$ is a valid witness for some $\mathbf t_i\in\pi$,
\[
\AVerify(\mu,\pi,\ASign(\mu,sk_i,\pi))=1
\]
with overwhelming probability. Its security has two aspects, signer ambiguity and unforgeability.
\end{definition}

\begin{definition}[Signer ambiguity / anonymity under full key exposure]\label{def:anon}
In $\Exp^{\mathsf{anon}}_{\Ambi,\mathcal A}$, the challenger runs $pp\gets\Setup(1^\kappa)$, $(\mathbf t_j,sk_j)\gets\AKeyGen(pp)$ for $j\in[N]$, and hands $\mathcal A$ all secret keys $\{sk_j\}$ (full key exposure). $\mathcal A$ outputs $(\mu^\ast,i_0,i_1,\pi^\ast)$ with $\mathbf t_{i_0},\mathbf t_{i_1}\in\pi^\ast\subseteq\{\mathbf t_j\}$; the challenger draws $b\us\{0,1\}$ and returns $\sigma^\ast\gets\ASign(\mu^\ast,sk_{i_b},\pi^\ast)$; $\mathcal A$ outputs $b'$. We write $\Advanon_{\Ambi}(\mathcal A)=\bigl\lvert\Pr[b'=1\mid b=1]-\Pr[b'=1\mid b=0]\bigr\rvert=2\,\bigl\lvert\Pr[b'=b]-\tfrac12\bigr\rvert$ (the distinguishing form of the advantage, which the game hops of Section \ref{sec:security} consume without loss of constants) and say $\Ambi$ is signer-ambiguous if this is negligible for every PPT $\mathcal A$.
\end{definition}

Ambiguity under full key exposure, the adversary knows every witness, is precisely what a one-out-of-many zero-knowledge proof provides, since the accepted transcript is (statistically or computationally) independent of the witness used. It is exactly the strength we need: in the withdrawability and unforgeability games of Section \ref{sec:security} the adversary may hold ring members' secret keys. The two-key left-or-right game above is equivalent, up to a factor equal to the ring size, to the ``guess the signer'' formulation of \cite{liu2024}, whose signer-ambiguity advantage is bounded by $(\text{ring size})^{-1}+\negl$.

\begin{definition}[Unforgeability of $\Ambi$]\label{def:ambiunf}
In $\Exp^{\mathsf{unf}}_{\Ambi,\mathcal A}$, the challenger generates $\{(\mathbf t_j,sk_j)\}_{j\in[N]}$, gives $\{\mathbf t_j\}$ to $\mathcal A$, and answers signing queries $\ASign(\mu,sk_i,\pi)$ (with $\pi\subseteq\{\mathbf t_j\}$ and signer index $i$), recording $(\mu,\pi)$ in $Q$; $\mathcal A$ may corrupt a strict subset $CO\subsetneq[N]$, receiving those $sk_j$. $\mathcal A$ wins by outputting $(\mu^\ast,\pi^\ast,\sigma^\ast)$ with $\AVerify(\mu^\ast,\pi^\ast,\sigma^\ast)=1$, $(\mu^\ast,\pi^\ast)\notin Q$, and $\pi^\ast\cap\{\mathbf t_j:j\notin CO\}\ne\emptyset$ (the ring contains an uncorrupted member). $\Advunf_{\Ambi}(\mathcal A)=\Pr[\mathcal A\text{ wins}]$.
\end{definition}

\subsection{Instantiation}\label{sec:ambiinstantiation}
We instantiate $\Ambi$ by an established lattice one-out-of-many proof \cite{groth2015,esgin2019}: a zero-knowledge proof that one of the public keys in $\pi$ opens to a short witness, made non-interactive by Fiat-Shamir. Its signer ambiguity (Definition \ref{def:anon}) is the honest-verifier zero-knowledge of the underlying $\Sigma$-protocol, which holds under full key exposure, as it uses no witness, and its unforgeability (Definition \ref{def:ambiunf}) is the knowledge soundness of the proof, reducing to MSIS. Such proofs achieve size logarithmic in the ring, and by Proposition \ref{prop:notranscription} they, rather than a naive challenge-split, are the correct lattice realisation. Only signer ambiguity is invoked in the proofs of the main construction (Section \ref{sec:security}); Definition \ref{def:ambiunf} is recorded for completeness and matches the treatment of \cite{liu2024}.

In particular, the one-out-of-many proofs of \cite{groth2015,esgin2019} are anonymous under full key exposure via the special honest-verifier zero-knowledge of their $\Sigma$-protocol, so Definition \ref{def:anon} is met concretely rather than merely assumed.

\section{Extended withdrawable signatures: syntax and security}\label{sec:ews}
Our security model is that of Liu, Susilo and Baek \cite{liu2024}: the same seven-algorithm syntax and the same three goals, correctness, unforgeability under insider corruption, and extended withdrawability. We refine it in two respects, both stated and justified in Section \ref{sec:relation}: we replace their completeness-style definition of extended withdrawability by an indistinguishability game (Definition \ref{def:wd}), and we add an explicit claimability-soundness notion (Definition \ref{def:claim}). Section \ref{sec:relation} then compares every definition, assumption and security statement with its counterpart in \cite{liu2024}, and justifies each discrepancy. Throughout, $N\ge 2$ is the number of potential signers and $i$ the index of the true signer.

\begin{definition}[Syntax]\label{def:syntax}
$\EWS=(\Setup,\KeyGen,\Shuffle,\WSign,\WSVerify,\Confirm,\CVerify,\Chk)$ with $pp\gets\Setup(1^\kappa)$, $(pk,sk)\gets\KeyGen(pp)$, $aux\gets\Shuffle(pk,\pi)=\pi\setminus\{pk\}$, $\sigma\gets\WSign(\mu,sk,aux)$, $1/0\gets\WSVerify(\mu,pk,\sigma)$ (public), $\delta\gets\Confirm(\mu,sk,pk,aux,\sigma)$, $1/0\gets\CVerify(\mu,pk,\sigma,\delta)$, $1/0\gets\Chk(\mu,pk,\sigma,\delta)$.
\end{definition}

\begin{definition}[Correctness]\label{def:correctness}
For all $\kappa$, all $\pi\ni pk$ and all $\mu$, running $\Setup,\KeyGen,\Shuffle,\WSign,\Confirm$ honestly yields $\WSVerify(\mu,pk,\sigma)=\CVerify(\mu,pk,\sigma,\delta)=\Chk(\mu,pk,\sigma,\delta)=1$ with overwhelming probability.
\end{definition}

\begin{definition}[Extended withdrawability: anonymity-until-claim]\label{def:wd}
Consider $\Exp^{\mathsf{wd}}_{\EWS,\mathcal A}(1^\kappa)$:
\begin{enumerate}
    \item $pp\gets\Setup(1^\kappa)$; $(pk_j,sk_j)\gets\KeyGen(pp)$ for $j\in[N]$ with $N\ge 2$; $\pi=\{pk_j\}_{j\in[N]}$.
    \item $\mathcal A$ is given $pp$, $\pi$ and all secret keys $\{sk_j\}_{j\in[N]}$, and outputs $(\mu^\ast,i_0,i_1,st)$ with $i_0\ne i_1$.
    \item the challenger draws $b\us\{0,1\}$ and $\rho^\ast\us S^m$, computes $(\sigma^\ast_r,C^\ast,\cdot)\gets\WSign(\mu^\ast,sk_{i_b},\pi\setminus\{pk_{i_b}\})$, where $C^\ast=\Com_{ck}(i_b;\rho^\ast)$, and returns the challenge in the $b$-independent canonical packaging $\sigma^\ast=(\sigma^\ast_r,C^\ast,aux^\ast)$, $aux^\ast:=\pi\setminus\{pk_{i_0}\}$. (Re-packaging is lossless and public by re-attribution, Theorem \ref{thm:correct}; without this convention the component $aux=\pi\setminus\{pk_{i_b}\}$ would name the signer by omission and the game would be trivially winnable.)
    \item $\mathcal A(st,\sigma^\ast)$ outputs $b'$; the experiment returns $[b'=b]$.
\end{enumerate}

Since $\mathcal A$ holds every $sk_j$, it can compute $\WSign$ and $\Confirm$ itself for any signature it constructs; the challenge $\sigma^\ast$ is the only object it cannot confirm, as the challenge randomness $\rho^\ast$ and the bit $b$ are hidden. $\EWS$ is extended-withdrawable if $\Advwd_{\EWS}(\mathcal A):=\bigl\lvert\Pr[b'=b]-\tfrac12\bigr\rvert\le\negl(\kappa)$.
\end{definition}

\begin{definition}[Unforgeability under insider corruption]\label{def:eufcmaEWS}
Consider $\Exp^{\mathsf{EUF\text{-}CMA}}_{\EWS,\mathcal A}(1^\kappa)$: $pp\gets\Setup(1^\kappa)$; $(pk_j,sk_j)\gets\KeyGen(pp)$ for $j\in[N]$; a target index $i$ is fixed; $\mathcal A$ receives $pp$, $\{pk_j\}_{j\in[N]}$ and the non-target keys $\{sk_j\}_{j\ne i}$ (insider corruption). Sets $W,M\gets\emptyset$. $\mathcal A$ has access to two oracles, each requiring $pk_i\in\pi$ and $\lvert\pi\rvert\ge 2$:
\begin{itemize}
    \item $\Oracle{WSign}(\mu,\pi)$: run $\sigma\gets\WSign(\mu,sk_i,\pi\setminus\{pk_i\})$ with its internal randomness $\rho$; store $(\mu,\pi,\sigma,\rho)$ in $W$; return $\sigma$.
    \item $\Oracle{Confirm}(\mu,\pi,\sigma)$: if $(\mu,\pi,\sigma,\rho)\in W$, set $M\gets M\cup\{\mu\}$ and return $\Confirm$ recomputed from the stored $\rho$; else return $\bot$.
\end{itemize}

$\mathcal A$ outputs $(\mu^\ast,\pi^\ast,\sigma^\ast,\delta^\ast)$ and wins if $\mu^\ast\notin M$, $pk_i\in\pi^\ast$, and $\WSVerify(\mu^\ast,pk_i,\sigma^\ast)=\CVerify(\mu^\ast,pk_i,\sigma^\ast,\delta^\ast)=\Chk(\mu^\ast,pk_i,\sigma^\ast,\delta^\ast)=1$. $\AdvEUF{\EWS}(\mathcal A):=\Pr[\mathcal A\text{ wins}]$.
\end{definition}

\begin{remark}[Scope: existential vs.\ strong unforgeability]\label{rem:suf}
Definition \ref{def:eufcmaEWS} targets existential (fresh-message) unforgeability. Strong unforgeability follows verbatim from the same reduction (Theorem \ref{thm:unforge}) if the base signature $\DS$ is SUF-CMA, as the no-hint, full-$t$ Dilithium variant of Section \ref{sec:instantiation} is, via MSIS, because $\delta_s$ signs the injective encoding $\inner{\mu,\pi,\sigma}$, so any change to $(\pi,\sigma)$ or to $\delta_s$ yields a fresh base message/signature pair.
\end{remark}

\begin{definition}[Claimability soundness]\label{def:claim}
$\Advclaim_{\EWS}(\mathcal A)$ is the probability that $\mathcal A$ outputs $(\mu,\pi,\sigma,\delta,\delta')$ with $\Chk(\mu,\pi[i],\sigma,\delta)=\Chk(\mu,\pi[i'],\sigma,\delta')=1$ and $i\ne i'$ (the indices revealed in $\delta,\delta'$), i.e.\ two accepted confirmations that attribute $\sigma$ to distinct ring members (the positions being well defined by the canonical order of Section \ref{sec:notation}).
\end{definition}

\subsection{Relation to the Liu-Susilo-Baek model}\label{sec:relation}
We now compare each notion above with its counterpart in \cite{liu2024} and justify every discrepancy. The definitional comparison is summarised, alongside the construction-level differences, in the consolidated comparison of Section \ref{sec:comparison} (Table \ref{tab:comparison}); here we take each notion in turn.

\paragraph{Syntax (Definition \ref{def:syntax}).}
Both models use the same seven algorithms (Definition \ref{def:syntax}) with the same roles: $\Shuffle$ publishes $aux=\pi\setminus\{pk\}$, $\WSVerify$ is public, and $\Chk$ ties a confirmed signature to its withdrawable object. The only differences are in the argument lists of $\Confirm$ and $\CVerify$: \cite{liu2024} write $\Confirm(m,sk,\sigma)$ and $\CVerify(m,pk,\tilde\sigma)$ in their syntax (while silently passing $aux$ and $\sigma$ in the construction), whereas we make the dependence explicit, $\Confirm(\mu,sk,pk,aux,\sigma)$ and $\CVerify(\mu,pk,\sigma,\delta)$. \emph{Justification:} in our construction the confirmed signature is an ordinary signature on the whole withdrawable object $\inner{\mu,\pi,\sigma}$, so both $\Confirm$ and $\CVerify$ must read $\sigma$ (and $\pi$, reconstructed from $pk$ and $aux$); making this explicit keeps traceability intrinsic to the confirmed signature rather than deferring it entirely to $\Chk$.

\paragraph{Correctness (Definition \ref{def:correctness}).}
This is identical to \cite[Definition 1]{liu2024}: an honestly produced withdrawable signature passes $\WSVerify$, its confirmation passes $\CVerify$, and the pair passes $\Chk$, all with overwhelming probability. No discrepancy.

\paragraph{Unforgeability under insider corruption (Definition \ref{def:eufcmaEWS}).}
This matches \cite[Definition 2]{liu2024} in substance: the adversary has a corruption oracle, a withdrawable signing oracle and a confirmation oracle, may hold every ring member's key except the target's, and wins by confirming a fresh message under the target key with all three verifications passing. Two refinements. (i) \emph{Fixed target.} \cite{liu2024} generate a pool of keys and let the adversary forge for whichever uncorrupted index it chooses; we fix the target index $i$ in advance. This is without loss of generality up to a factor $N$ (a reduction guesses $i$; see Theorem \ref{thm:unforge}) and makes the reduction cleaner. (ii) \emph{Bookkeeping.} We record signed tuples in $W$ and let $\Oracle{Confirm}$ act only on $\sigma\in W$; \cite{liu2024} do the same (``if $\sigma\in W$''). A point of rigour rather than definition: \cite[Thm. 1]{liu2024} claims a reduction loss $L=q_H$ (the number of hash queries), whereas our reduction (Theorem \ref{thm:unforge}) incurs no multiplicative $q_H$ factor, it neither rewinds the adversary nor guesses a hash index, the only losses being the factor $N$ above and a $q_W\Advanon_{\Ambi}$ term from simulating the signing oracle. We view the tighter reduction as a benefit of the modular treatment.

\paragraph{Extended withdrawability (Definition \ref{def:wd}): the main discrepancy.}
Here we depart from \cite{liu2024} deliberately. Their Definition 3 declares the experiment successful when the adversary outputs $(pk',aux')$ with $aux'=\gamma\setminus\{pk'\}$ and $\WSVerify(m^\ast,pk',\sigma^\ast)=1$, and requires $\Pr[\Exp^{\mathsf{Withdraw}}=1]=1$. This is a completeness statement, it asserts that a withdrawable signature can always be re-attributed to another ring member and still verify, and it contains no hidden bit, hence no hiding guarantee. Read literally, it is met by any scheme whose $\WSVerify$ accepts, including one that plainly reveals the signer. \cite{liu2024} obtain the actual hiding elsewhere, from the signer ambiguity of the underlying one-out-of-$N$ signature (their Definition 4) together with the ``$1/n$-guess'' step in the proof of their Theorem 3. We instead fold both into a single anonymity-until-claim game (Definition \ref{def:wd}): a hidden bit $b$ selects which of two signer identities produces the challenge, the advantage $\lvert\Pr[b'=b]-\tfrac12\rvert$ must be negligible, and the ``until claim'' clause is enforced by withholding the challenge randomness $\rho^\ast$, so that the commitment $C^\ast$ cannot be opened and $\sigma^\ast$ cannot be confirmed. \emph{Justification:} (a) a completeness statement is not a security definition, so the hiding must be phrased as indistinguishability; (b) our game captures exactly the guarantee \cite{liu2024} intend, ambiguity among the potential signers until confirmation, now as one self-contained notion that the construction provably meets (Theorem \ref{thm:wd}), rather than a property split across a completeness claim and a building-block assumption; and (c) nothing is lost, since the re-attribution completeness that Definition 3 does capture is, in our treatment, an immediate consequence of correctness of ring verification ($\WSVerify$ accepts for every $pk'\in\pi$). This is the most consequential definitional change and the reason our withdrawability theorem is a reduction to commitment hiding and ring anonymity rather than a probability-one assertion.

\paragraph{Signer ambiguity of the building block (Definition \ref{def:anon}).}
The hiding that \cite{liu2024} place in their Definition 4 (signer ambiguity of the one-out-of-$N$ signature, advantage $\le\tfrac1N+\negl$) survives verbatim as our Definition \ref{def:anon}, stated as a two-key left-or-right game under full key exposure and equivalent to theirs up to the factor $N$. In our model it lives with the building block (Section \ref{sec:buildingblock}) and is invoked by the withdrawability theorem, rather than standing beside the withdrawability definition; this is the structural counterpart of folding Definition 3 and Definition 4 into one notion.

\paragraph{Claimability soundness (Definition \ref{def:claim}): a new notion.} \cite{liu2024} have no explicit counterpart. Their $\Chk$ links $\tilde\sigma$ to $\sigma$ (through the embedded group element $g^r$), and their unforgeability prevents a non-signer from confirming, but they never isolate the property that a confirmed signature binds to a unique signer index, that no adversary can present two valid confirmations of one $\sigma$ attributing it to different members. In our construction the signer's index is an explicit committed value, so this property is both meaningful and necessary for ``claim'' to be sound; we state it as Definition \ref{def:claim} and prove it (Proposition \ref{prop:claim}) from the binding of the commitment. \emph{Justification:} this is exactly the claimability of a claimable ring signature \cite{park2019}, the abstraction \cite{liu2024} themselves point to; making it explicit closes a gap in the original model.

\paragraph{Assumptions.}
The model of \cite{liu2024} is realised under the discrete-logarithm assumption (via the EUF-CMA security of the discrete-log signature and the ambiguity of the discrete-log one-out-of-$N$ signature). We replace these by their post-quantum, module-lattice analogues (Section \ref{sec:hardness}): decisional MLWE, MSIS and SelfTargetMSIS, together with the base signature's EUF-CMA and the ring signature's anonymity. Our claimable-ring design additionally uses a commitment, hence the extra hiding (MLWE) and binding (MSIS) assumptions of Section \ref{sec:commit}; \cite{liu2024} need no commitment because their confirmation links through the group element $g^r$, a re-randomisation that is free in a prime-order group but has no cheap lattice analogue (Section \ref{sec:comparison}).

\paragraph{Security statements and their proofs.}
Beyond the definitions, our results correspond to the three theorems of \cite{liu2024}, and our proofs close gaps in theirs.
\begin{itemize}
    \item Unforgeability: \cite[Thm. 1]{liu2024} reduces to the EUF-CMA security of the discrete-log signature but asserts a reduction loss $L=q_H$ with no rewinding or hash-guessing argument, and in extracting the forgery it identifies the confirmed component with the underlying signature. Our Theorem \ref{thm:unforge} gives an explicit reduction with no $q_H$ loss and a clean extraction: the confirmed signature $\delta_s$ on the fresh string $\inner{\mu^\ast,\pi^\ast,\sigma^\ast}$ is the base forgery.
    \item Signer ambiguity: \cite[Thm. 2]{liu2024} simulates the challenge by a uniform tuple $(t_1,\dots,t_N,z)$ and concludes advantage $1/2$ ``since $b,sk_b$ are unused''; the step that makes the simulated tuple verify, programming $H$ at the induced commitment, is left implicit. Over a field this is harmless (a uniform tuple can always be made to verify), but the same omission is fatal over lattices, where a uniform tuple neither lies in the sparse set $\Btau$ nor passes the rejection test (Proposition \ref{prop:notranscription}). We therefore take ambiguity to be the zero-knowledge of a one-out-of-$N$ proof (Definition \ref{def:anon}), made rigorous with the programming step explicit.
    \item Extended withdrawability: \cite[Thm. 3]{liu2024} argues that, ambiguity holding, the adversary guesses the signer with probability $1/N$ and that $\WSVerify$ still accepts, a mixture of an informal $1/N$ bound and the completeness of re-attribution. Our Theorem \ref{thm:wd} is a game-based reduction to commitment hiding and ring anonymity, and we additionally prove claimability soundness (Proposition \ref{prop:claim}), for which \cite{liu2024} state no result.
    \item A new statement: Proposition \ref{prop:notranscription} (no secure literal lattice transcription) has no counterpart in \cite{liu2024}: it is a lattice-specific impossibility that justifies treating the one-out-of-$N$ signature as a black box rather than transcribing their concrete scheme.
\end{itemize}

\paragraph{Relation to the earlier lattice transcription.}
A direct lattice transcription of \cite{liu2024} adopts their definitions verbatim, including the completeness-style Definition 3 and the informal ambiguity argument, and so inherits the gaps above, compounded by two lattice-specific defects: it publishes a signer-derived shift $\{A\mathbf s_1 s_j\}$ that breaks ambiguity (Section \ref{sec:comparison}), and its one-out-of-$N$ layer is the unsound challenge-split of Proposition \ref{prop:notranscription}. The refinements of this section, the hiding game (Definition \ref{def:wd}), the explicit claimability notion (Definition \ref{def:claim}), and the black-box ambiguity assumption (Definition \ref{def:anon}), correct issues present in both \cite{liu2024} and that transcription, and are exactly what the construction of Section \ref{sec:construction} is designed to satisfy.

\section{The construction}\label{sec:construction}

\subsection{Ingredients and random oracles}\label{sec:ingredients}
The construction combines the three primitives fixed in Sections \ref{sec:sig}--\ref{sec:commit}.

\begin{enumerate}
    \item A Fiat-Shamir with aborts base signature $\DS=(\KeyGen,\Sign,\Verify)$ over the common matrix $A$ (Section \ref{sec:fswa}), whose signing algorithm evaluates a challenge random oracle $H_{\DS}:\{0,1\}^\ast\to\Btau$.
    \item An anonymous one-out-of-$N$ signature $\Ambi=(\AKeyGen,\ASign,\AVerify)$ (Section \ref{sec:buildingblock}), whose non-interactive proof evaluates its own challenge oracle $H_{\Ambi}:\{0,1\}^\ast\to\mathcal X$, where $\mathcal X$ is the challenge space of the underlying $\Sigma$-protocol (a bounded subset of $\Rq$ for a one-out-of-many proof; any finite set for the abstract interface). We take $\AKeyGen=\KeyGen$, so ring members and signers share the key format $pk=\mathbf t=A\mathbf s_1+\mathbf s_2$, $sk=(\mathbf s_1,\mathbf s_2)$.
    \item A commitment $\Com_{ck}:\mathcal M\times S^m\to\Rq^{k'}$ with key $ck=\bar A\in\Rq^{k'\times m}$ (Section \ref{sec:commit}), together with an injective index encoding $\Encode:[N]\to\mathcal M$; we write $\Com_{ck}(i;\rho)$ for $\Com_{ck}(\Encode(i);\rho)$, where $i$ is the signer's position in $\pi$ under the canonical order of Section \ref{sec:notation}, so that ``the committed index'' names a unique ring member.
\end{enumerate}

Here $\inner{\cdot}:(\{0,1\}^\ast)^\ast\to\{0,1\}^\ast$ is the fixed injective tuple encoding of Section \ref{sec:notation}. The two random oracles $H_{\DS}$ and $H_{\Ambi}$ are domain-separated (say, by a distinct one-bit prefix), so a query to one is never a query to the other; the extended withdrawable signature introduces no further oracle of its own, invoking $H_{\DS},H_{\Ambi}$ only through $\Sign/\Verify$ and $\ASign/\AVerify$. Correctness and security are stated in the (quantum) random-oracle model (Section \ref{sec:qrom}).

\subsection{Setup and key generation}\label{sec:setupkeygen}

\begin{algorithm}[h]
\caption{$\Setup(1^\kappa)$ and $\KeyGen(pp)$}\label{alg:setup}
\begin{algorithmic}[1]
\Procedure{$\Setup$}{$1^\kappa$}
    \State $A\us\Rq^{k\times l}$;\quad $ck=\bar A\us\Rq^{k'\times m}$;\quad \Return $pp=(A,ck)$
\EndProcedure
\Procedure{$\KeyGen$}{$pp$}
    \State $(\mathbf s_1,\mathbf s_2)\us S^l_\eta\times S^k_\eta$;\quad $\mathbf t=A\mathbf s_1+\mathbf s_2$;\quad
    \Return $pk=\mathbf t$, $sk=(\mathbf s_1,\mathbf s_2)$
\EndProcedure
\end{algorithmic}
\end{algorithm}

$\Setup$ fixes the parameters shared by all users: a uniform module matrix $A\us\Rq^{k\times l}$ and a commitment key $\bar A\us\Rq^{k'\times m}$. A single common $A$ is essential (Section \ref{sec:pp}): each public key is the module-LWE sample $\mathbf t=A\mathbf s_1+\mathbf s_2\in\Rq^k$ taken against the same $A$, so that the membership relation ``$\exists$ short $(\mathbf s_1,\mathbf s_2):\mathbf t_j=A\mathbf s_1+\mathbf s_2$'' proved by the one-out-of-$N$ layer is identical for every ring member, which is what makes ``one of the $\mathbf t_j$ opens to a short witness'' a well-formed statement across $\pi$. The secret key is the short pair $(\mathbf s_1,\mathbf s_2)\in S^l_\eta\times S^k_\eta$; its shortness is what the base signature's rejection sampling exploits and what decisional MLWE (Definition \ref{def:mlwe}) protects against key recovery. Finally, $\Setup$ is a public-coin common reference string: $A$ and $\bar A$ are uniformly sampled and no party holds, or needs, a trapdoor, since binding (MSIS) and hiding (MLWE) hold for uniformly generated $\bar A$.

The contrast with the group-based $\KeyGen$ of \cite{liu2024} is discussed in Section \ref{sec:compalg}.

\subsection{Withdrawable signing and public verification}\label{sec:wsign}

\begin{algorithm}[h]
\caption{$\WSign(\mu,sk_i,aux)$ and $\WSVerify(\mu,pk,\sigma)$}\label{alg:wsign}
\begin{algorithmic}[1]
\Procedure{$\WSign$}{$\mu,sk_i,aux$}
    \State $\pi=aux\cup\{pk_i\}$;\quad $\rho\us S^m$;\quad $C=\Com_{ck}(i;\rho)$ \Comment{hiding commitment to signer index}
    \State $\sigma_r\gets\ASign(\inner{\mu,C},sk_i,\pi)$ \Comment{$C$ bound into the ring message}
    \State \Return $\sigma=(\sigma_r,C,aux)$
\EndProcedure
\Procedure{$\WSVerify$}{$\mu,pk,\sigma$} \Comment{public}
    \State parse $(\sigma_r,C,aux)$;\quad $\pi=aux\cup\{pk\}$;\quad \Return $\AVerify(\inner{\mu,C},\pi,\sigma_r)$
\EndProcedure
\end{algorithmic}
\end{algorithm}

$\WSign$ proceeds in two steps. First it \emph{commits to its own index}: it samples short randomness $\rho\us S^m$ and sets $C=\Com_{ck}(i;\rho)\in\Rq^{k'}$, a commitment that hides $i$ under decisional MLWE (Definition \ref{def:hiding}) and to which the signer will later be bound (Definition \ref{def:binding}). Second, it produces a one-out-of-$N$ signature on the message $\inner{\mu,C}$ over the ring $\pi=aux\cup\{pk_i\}$ using its own witness, $\sigma_r\gets\ASign(\inner{\mu,C},sk_i,\pi)$. Folding $C$ into the signed message binds the commitment to the transcript through $H_{\Ambi}$, so $C$ cannot be stripped or replaced without invalidating $\sigma_r$; this is what makes the later claim sound (Remark \ref{rem:plainsig}). The withdrawable signature is $\sigma=(\sigma_r,C,aux)$. Verification is public: $\WSVerify$ reconstructs $\pi=aux\cup\{pk\}$ and returns $\AVerify(\inner{\mu,C},\pi,\sigma_r)$, reading only public keys.

Two properties are immediate from the algorithm. \emph{Ambiguity:} $\sigma_r$ hides the signer index by the anonymity of $\Ambi$ (Definition \ref{def:anon}) and $C$ hides it by commitment hiding, so $\sigma$ leaks nothing about which member of $\pi$ signed. \emph{Extended withdrawability (re-attribution):} for any $pk'\in\pi$ the repackaged signature $\sigma'=(\sigma_r,C,\pi\setminus\{pk'\})$ verifies, $\WSVerify(\mu,pk',\sigma')=1$, since $(\pi\setminus\{pk'\})\cup\{pk'\}=\pi$; one withdrawable signature verifies under every potential signer. This is exactly the completeness that Definition 3 of \cite{liu2024} isolates, and which our hiding game (Definition \ref{def:wd}) subsumes. The corresponding change to $\WSign$, a committed index in place of published shifts, is compared in Section \ref{sec:compalg}.

\subsection{Confirmation, confirmed verification and checking}\label{sec:confirm}

\begin{algorithm}[h]
\caption{$\Confirm$, $\CVerify$, $\Chk$}\label{alg:confirm}
\begin{algorithmic}[1]
\Procedure{$\Confirm$}{$\mu,sk_i,pk_i,aux,\sigma$}
    \State $\pi=aux\cup\{pk_i\}$;\quad $\delta_s\gets\Sign_{sk_i}(\inner{\mu,\pi,\sigma})$;\quad \Return $\delta=(\delta_s,i,\rho)$ \Comment{$(i,\rho)$: position and opening retained by the signer from $\WSign$}
\EndProcedure
\Procedure{$\CVerify$}{$\mu,pk_i,\sigma,\delta$}
    \State parse $(\delta_s,i,\rho)$;\quad \Return $\Verify_{pk_i}(\inner{\mu,\pi,\sigma},\delta_s)$
\EndProcedure
\Procedure{$\Chk$}{$\mu,pk_i,\sigma,\delta$}
    \State parse $\sigma=(\sigma_r,C,aux)$, $\delta=(\delta_s,i,\rho)$
    \If{$pk_i=(aux\cup\{pk_i\})[i]\ \wedge\ C=\Com_{ck}(i;\rho)\ \wedge\ \WSVerify(\mu,pk_i,\sigma)=1\ \wedge\ \CVerify(\mu,pk_i,\sigma,\delta)=1$}
    \State \Return $1$
    \EndIf
    \State \Return $0$
\EndProcedure
\end{algorithmic}
\end{algorithm}

$\Confirm$ is the signer’s claim, and it is stateful in one respect: per signature, the signer retains the opening randomness $\rho$ drawn inside $\WSign$ together with its position $i$ in $\pi$; the experiments of Section \ref{sec:ews} model exactly this by storing $\rho$ in $W$. Holding $sk_i$ and $(i,\rho)$, the signer produces an ordinary base signature on the entire withdrawable object, $\delta_s\gets\Sign_{sk_i}(\inner{\mu,\pi,\sigma})$, and reveals the opening $(i,\rho)$ of the index commitment $C$; the confirmed signature is $\delta=(\delta_s,i,\rho)$. Signing $\inner{\mu,\pi,\sigma}$ rather than $\mu$ alone makes $\delta_s$ inseparable from this particular $\sigma$, and revealing $(i,\rho)$ lifts the ambiguity by naming the signer. $\CVerify$ is public: it returns $\Verify_{pk_i}(\inner{\mu,\pi,\sigma},\delta_s)$, so a confirmed signature is checkable under the signer's public key by anyone. $\Chk$ certifies traceability on top of $\CVerify$: it recomputes $C=\Com_{ck}(i;\rho)$ (the revealed index is the one committed inside $\sigma$) and requires $\WSVerify(\mu,pk_i,\sigma)=1$ as well; hence $\Chk$ accepts only when $\delta$ confirms this $\sigma$ and attributes it to the unique committed signer.

\begin{remark}[Why a plain signature is not a sound claim]\label{rem:plainsig}
$\delta_s=\Sign_{sk_i}(\inner{\mu,\pi,\sigma})$ proves the claimant holds $sk_i$, but every ring member holds its own secret key, so on its own any member could ``confirm'' any $\sigma$. The commitment $C$, created at signing time and bound into the ring message $\inner{\mu,C}$, is fixed by $\sigma$; by binding it opens to a single index, the true signer's. $\Chk$ accepts only if the opened $i$ matches $C$, so only the actual author can confirm.
\end{remark}

\begin{remark}[Reuse of the key across the two layers]\label{rem:keyreuse}
The same secret $sk=(\mathbf s_1,\mathbf s_2)$ serves as the $\Ambi$ membership witness in $\sigma_r$ and as the $\DS$ signing key in $\delta_s$. This is safe because both primitives are zero-knowledge: the accepted $\Ambi$ transcript is witness-independent (Definition \ref{def:anon}) and $\DS$ is naHVZK (Section \ref{sec:fswa}), so a distinguisher's joint view of $\sigma_r$ and any confirmed $\delta_s$ under $pk_i$ is efficiently simulatable and leaks nothing about $\mathbf s_1$ beyond $pk_i=\mathbf t_i$. Concretely, this is why Step 1 of Theorem \ref{thm:unforge} may answer confirmations with the real key while the reduction of Step 2 never needs $sk_i$: the two uses of the key are never simultaneously required under the target.
\end{remark}

Our $\Confirm$ replaces their two signatures and embedded group element by one base signature and one commitment opening; see Section \ref{sec:compalg}.

\section{Security analysis}\label{sec:security}

We establish, in the (quantum) random-oracle model, the four properties of Section \ref{sec:ews}: correctness (Theorem \ref{thm:correct}), extended withdrawability (Theorem \ref{thm:wd}), unforgeability under insider corruption (Theorem \ref{thm:unforge}) and claimability soundness (Proposition \ref{prop:claim}). Each result reduces the security of the composed scheme to a single, previously stated guarantee of an ingredient: correctness to the correctness of $\DS,\Ambi,\Com$; extended withdrawability to commitment hiding (hence decisional MLWE) and the anonymity of $\Ambi$; unforgeability to the EUF-CMA security of the base signature $\DS$; and claimability soundness to commitment binding (hence MSIS). All reductions are tight up to the stated additive terms and the factor $N$ from guessing the target signer, and none rewinds the adversary. Throughout we write $\sigma=(\sigma_r,C,aux)$, $\delta=(\delta_s,i,\rho)$ and $\pi=aux\cup\{pk_i\}$ as in Section \ref{sec:construction}.

Correctness requires that an honestly generated withdrawable signature and its confirmation pass all three verification algorithms. Because the construction is modular, this follows from the correctness of the three ingredients; the lattice rounding that makes it hold at the concrete level is confined to $\DS$ and $\Ambi$ and rests on the rounding-stability Lemma \ref{lem:rounding}.

\begin{theorem}[Correctness]\label{thm:correct}
If $\DS$, $\Ambi$ and $\Com$ are correct, then $\EWS$ is correct (Definition \ref{def:correctness}).
\end{theorem}

\begin{proof}
Fix $\kappa$, a ring $\pi$ with $pk_i\in\pi$, a message $\mu$, and honestly generated $pp,\{(pk_j,sk_j)\}$; let $\sigma=(\sigma_r,C,aux)\gets\WSign(\mu,sk_i,aux)$ and $\delta=(\delta_s,i,\rho)\gets\Confirm(\mu,sk_i,pk_i,aux,\sigma)$. We check the three requirements of Definition \ref{def:correctness}.

\begin{enumerate}
    \item $\WSVerify(\mu,pk_i,\sigma)=1$. By construction $\sigma_r=\ASign(\inner{\mu,C},sk_i,\pi)$ with $sk_i$ a valid witness for $pk_i=\mathbf t_i\in\pi$, and $\WSVerify$ returns $\AVerify(\inner{\mu,C},\pi,\sigma_r)$, which is $1$ with overwhelming probability by correctness of $\Ambi$ (Definition \ref{def:oneofn}). At the concrete level (Section \ref{sec:instantiation}) this is the completeness of the one-out-of-many $\Sigma$-protocol made non-interactive by Fiat-Shamir: an honestly generated membership proof verifies with overwhelming probability. (The $\HighBits$/$\LowBits$ rounding of Lemma \ref{lem:rounding} enters only through the base signature $\DS$ in $\CVerify$/$\Chk$, not through $\Ambi$.)
    \item $\CVerify(\mu,pk_i,\sigma,\delta)=1$. Here $\delta_s=\Sign_{sk_i}(\inner{\mu,\pi,\sigma})$ and $\CVerify$ returns $\Verify_{pk_i}(\inner{\mu,\pi,\sigma},\delta_s)$, which is $1$ with overwhelming probability by correctness of $\DS$ (Definition \ref{def:sig}).
    \item $\Chk(\mu,pk_i,\sigma,\delta)=1$. $\Chk$ recomputes $\Com_{ck}(i;\rho)$ and compares it with the $C$ contained in $\sigma$; the two are equal because $\Confirm$ reveals the very randomness $\rho$ used to form $C$ at signing time. The position check $pk_i=(aux\cup\{pk_i\})[i]$ holds by construction, since the honest $\Confirm$ reveals exactly the signer's position in $\pi$. Together with items 1 and 2, which $\Chk$ re-invokes, all four conjuncts hold and $\Chk$ returns $1$.
\end{enumerate}

Each item fails only with the negligible correctness error of the corresponding ingredient, so their conjunction holds with overwhelming probability. Finally, re-attribution is a corollary of item 1: for any $pk'\in\pi$, the repackaged signature $\sigma'=(\sigma_r,C,\pi\setminus\{pk'\})$ satisfies $\WSVerify(\mu,pk',\sigma')=1$, since $(\pi\setminus\{pk'\})\cup\{pk'\}=\pi$ reconstructs the same ring and message $\inner{\mu,C}$, the completeness that Definition 3 of \cite{liu2024} isolates and that our hiding game subsumes.
\end{proof}

Extended withdrawability is the hiding guarantee: until the signer confirms, no efficient adversary, even one holding every secret key in the ring, can tell which of two nominated members produced a withdrawable signature. Exactly two components of $\sigma$ could betray the signer: the index commitment $C$ and the ring signature $\sigma_r$. The proof neutralises them in turn by a three-game sequence, first replacing $C$ by a commitment to a fixed dummy index (commitment hiding), then re-signing $\sigma_r$ with a fixed witness (ring anonymity), after which the challenge is manifestly independent of the hidden bit $b$. Crucially, since $\mathcal A$ holds all keys it answers its own $\WSign$ and $\Confirm$ queries, so the reduction alters only the challenge $\sigma^\ast$; and $\mathcal A$ cannot confirm $\sigma^\ast$, because the challenge randomness $\rho^\ast$ (which would open $C^\ast$) is withheld, this is the ``until claim'' clause.

\begin{theorem}[Extended withdrawability]\label{thm:wd}
For every PPT $\mathcal A$ there are PPT $\mathcal D,\mathcal B$ with
\[
\Advwd_{\EWS}(\mathcal A)\le\Advhide_{\Com}(\mathcal D)+\Advanon_{\Ambi}(\mathcal B)\le \AdvMLWE(\mathcal D)+\Advanon_{\Ambi}(\mathcal B).
\]
\end{theorem}

\begin{proof}
Write $p_x:=\Pr[b'=b\mid G_x]$ for the winning probability in game $G_x$. The adversary already holds every $sk_j$ and constructs any auxiliary signatures itself, so the challenger’s only action is to produce the challenge $\sigma^\ast=(\sigma^\ast_r,C^\ast, aux^\ast)$; the three games differ only in how $C^\ast$ and $\sigma^\ast_r$ are formed. In every game the packaging is the canonical $aux^\ast=\pi\setminus\{pk_{i_0}\}$ of Definition \ref{def:wd}, so the $aux$ component carries no information about $b$. Recall $\sigma^\ast_r=\ASign(\inner{\mu^\ast,C^\ast},sk_{u},\pi)$ for a witness index $u$ and $C^\ast=\Com_{ck}(v;\rho^\ast)$ for a committed index $v$; the three games vary the pair $(u,v)$.

\begin{itemize}
    \item Game $G_0$ (real): $C^\ast=\Com_{ck}(i_b;\rho^\ast)$, $\sigma^\ast_r=\ASign(\inner{\mu^\ast,C^\ast},sk_{i_b},\pi)$. By definition: 
    \[
    p_0=\tfrac12\pm\Advwd_{\EWS}(\mathcal A), \text{so } \Advwd_{\EWS}(\mathcal A)=\lvert p_0-\tfrac12\rvert.
    \]

    \item Game $G_1$ (dummy commitment): Identical to $G_0$ except $C^\ast=\Com_{ck}(i_0;\rho^\ast)$ (a commitment to the fixed index $i_0$, independent of $b$); the ring signature is still 
    \[
    \sigma^\ast_r=\ASign(\inner{\mu^\ast,C^\ast},sk_{i_b},\pi).
    \]
    
    \emph{Bound $\lvert p_0-p_1\rvert$.} Let $\mathcal D$ be a hiding adversary in the sense of Definition \ref{def:hiding}: it names two messages $(x_0,x_1)$ and receives $C^\ast=\Com_{ck}(x_\beta;\rho^\ast)$ with fresh hidden $\rho^\ast$, for one of the two cases $\beta\in\{0,1\}$; its advantage is the gap between its acceptance probabilities in the two cases. $\mathcal D$ generates $pp$ (drawing $ck$ from the hiding challenger and $A$ itself) and all $(pk_j,sk_j)$, hands $\mathcal A$ everything, and receives $(\mu^\ast,i_0,i_1,st)$. It picks $b\us\{0,1\}$, submits the single query $(m_0,m_1)=(i_b,i_0)$, receives $C^\ast$, computes $\sigma^\ast_r=\ASign(\inner{\mu^\ast,C^\ast},sk_{i_b},\pi)$ (it knows $sk_{i_b}$), returns $\sigma^\ast=(\sigma^\ast_r,C^\ast,aux^\ast)$, and outputs $[b'=b]$. If $\beta=0$ then $C^\ast=\Com_{ck}(i_b;\cdot)$ and $\mathcal D$ perfectly simulates $G_0$; if $\beta=1$ then $C^\ast=\Com_{ck}(i_0;\cdot)$ and it simulates $G_1$. Hence $\lvert p_0-p_1\rvert\le\Advhide_{\Com}(\mathcal D)\le\AdvMLWE(\mathcal D)$.

    \item Game $G_2$ (fixed witness): Identical to $G_1$ except the ring signature uses the fixed witness $sk_{i_0}$: $\sigma^\ast_r=\ASign(\inner{\mu^\ast,C^\ast},sk_{i_0},\pi)$. Now none of $C^\ast$ (as of $G_1$), $\sigma^\ast_r$, or the canonical $aux^\ast$ depends on $b$, so the whole view of $\mathcal A$ is independent of $b$; therefore $p_2=\tfrac12$.
    
    \emph{Bound $\lvert p_1-p_2\rvert$.} Let $\mathcal B$ play the $\Ambi$ anonymity game (Definition \ref{def:anon}), whose challenger provides all $\{sk_j\}$ and, on a query $(m,j_0,j_1,\pi)$, returns $\ASign(m,sk_{j_{\beta'}},\pi)$ for its bit $\beta'$. $\mathcal B$ receives the keys, hands them and $pp,\pi$ to $\mathcal A$ (relaying $H_{\Ambi}$ to its challenger's oracle and lazily sampling $H_{\DS}$ itself), gets $(\mu^\ast,i_0, i_1,st)$, picks $b\us\{0,1\}$ and $\rho^\ast$, forms $C^\ast=\Com_{ck}(i_0;\rho^\ast)$, and submits the anonymity query $(\inner{\mu^\ast,C^\ast},i_b,i_0,\pi)$; it plants the returned $\sigma^\ast_r$ into $\sigma^\ast$ and outputs $[b'=b]$. If $\beta'=0$ the witness is $sk_{i_b}$ (game $G_1$); if $\beta'=1$ the witness is $sk_{i_0}$ (game $G_2$). Full key exposure makes it legitimate for $\mathcal B$ to hold and forward all secret keys. Hence $\lvert p_1-p_2\rvert\le\Advanon_{\Ambi}(\mathcal B)$.
\end{itemize}

Combining the three steps,
\[
\Advwd_{\EWS}(\mathcal A)=\Bigl\lvert p_0-\tfrac12\Bigr\rvert=\lvert p_0-p_2\rvert\le\lvert p_0-p_1\rvert+\lvert p_1-p_2\rvert\le\Advhide_{\Com}(\mathcal D)+\Advanon_{\Ambi}(\mathcal B).\qedhere
\]
\end{proof}

Unforgeability under insider corruption guarantees that only the genuine signer can confirm: no adversary, even one holding every ring member's key except the target's, can output a confirmed signature on a message the target never confirmed. Note the withdrawable signature $\sigma_r$ is itself a ring signature that the adversary could forge with a corrupted witness, so the unforgeable core is the confirmation, where $\delta_s$ is a base signature under the target key $pk_i$ on the whole object $\inner{\mu,\pi,\sigma}$. Accordingly the reduction plants the EUF-CMA challenge at $pk_i$: it answers withdrawable-signing queries without $sk_i$ by ring-signing with a corrupted witness (indistinguishable by anonymity, Step 1) and confirmation queries through its own signing oracle, so that a confirmed forgery on a fresh message is exactly a base-scheme forgery (Step 2). The reduction never rewinds and incurs no multiplicative $q_H$ factor, the point of rigour on which it improves
\cite[Thm. 1]{liu2024}.

\begin{theorem}[Unforgeability under insider corruption]\label{thm:unforge}
For every PPT $\mathcal A$ making $q_W$ signing and $q_C$ confirmation queries there are PPT $\mathcal B,\mathcal B'$ with
\[
\AdvEUF{\EWS}(\mathcal A)\le\AdvEUF{\DS}(\mathcal B)+q_W\,\Advanon_{\Ambi}(\mathcal B').
\]
$\mathcal B$ makes at most $q_C$ signing queries and does not rewind $\mathcal A$ (no multiplicative $q_H$ loss).
\end{theorem}

\begin{proof}
For a ring $\pi\ni pk_i$ with $\lvert\pi\rvert\ge 2$ write $j_0(\pi)$ for a fixed member of $\pi\setminus\{pk_i\}$, a witness $sk_{j_0(\pi)}$ for it is known to the honest keyholder (all non-target keys are honestly generated), and $\iota(\pi)$ for the position of $pk_i$ in $\pi$ under the canonical order (Section \ref{sec:notation}), i.e.\ the index the honest signer commits to and reveals. We proceed in two steps: a game hop that removes the use of $sk_i$ from the signing oracle, and a reduction that turns a forgery into a $\DS$ forgery.

\begin{itemize}
    \item Step 1: replacing the signing witness: Let $H_0=\Exp^{\mathsf{EUF\text{-}CMA}}_{\EWS,\mathcal A}$ be the real game (the $\Oracle{WSign}$ oracle signs the ring with witness $sk_i$). Let $H_1$ be identical except that $\Oracle{WSign}(\mu,\pi)$ answers with $\sigma=(\ASign(\inner{\mu,C},sk_{j_0(\pi)},\pi),C,\pi\setminus\{pk_i\})$, $C=\Com_{ck}(\iota(\pi);\rho)$, the same commitment $C$ to the position $\iota(\pi)$ and the same message $\inner{\mu,C}$, but the ring is signed with the alternative witness $sk_{j_0(\pi)}$. All other steps (in particular $C$, $W$, $\Oracle{Confirm}$ and the win predicate) are unchanged. We claim
    \[
    \bigl\lvert\Pr[\mathcal A\text{ wins }H_0]-\Pr[\mathcal A\text{ wins }H_1]\bigr\rvert\le q_W\,\Advanon_{\Ambi}(\mathcal B').
    \]
    Define hybrids $H^{(0)},\dots,H^{(q_W)}$ where $H^{(t)}$ answers the first $t$ $\Oracle{WSign}$ queries with the alternative witness $sk_{j_0(\pi)}$ and the remaining ones with $sk_i$; then $H^{(0)}=H_0$ and $H^{(q_W)}=H_1$. A distinguisher $\mathcal B'$ against $\Ambi$-anonymity (Definition \ref{def:anon}) receives all $\{sk_j\}$ from its challenger, legitimate, since anonymity is under full key exposure, and this is exactly why $\mathcal B'$ may hold $sk_i$ here even though $\mathcal B$ below will not, picks $t^\ast\us[q_W]$, and simulates $\mathcal A$ answering the $s$-th $\Oracle{WSign}$ query with $sk_{j_0}$ if $s<t^\ast$, with $sk_i$ if $s>t^\ast$, and, for $s=t^\ast$ on ring $\pi$ and message $\inner{\mu,C}$, by submitting the anonymity query $(\inner{\mu,C},i,j_0(\pi),\pi)$ and planting the reply. It relays $H_{\Ambi}$ to its anonymity challenger's oracle, so the planted challenge verifies in $\mathcal A$'s view, and simulates $H_{\DS}$ by lazy sampling; it answers $\Oracle{Confirm}$ honestly (it holds all keys and all stored $\rho$) and finally outputs $[\mathcal A\text{ wins}]$. When the anonymity bit selects witness $sk_i$ (resp.\ $sk_{j_0}$) the $t^\ast$-th answer matches $H^{(t^\ast-1)}$ (resp.\ $H^{(t^\ast)}$), so a standard telescoping over the uniform $t^\ast$ yields the displayed bound.

    \item Step 2: reduction to $\DS$ unforgeability: $\mathcal B$ receives from the $\DS$ challenger the public matrix $A$, an EUF-CMA challenge key $pk^\ast$ generated under it, a signing oracle $\Oracle{Sign}(\cdot)$ and the random oracle $H_{\DS}$; it completes the public parameters as $pp=(A,ck)$ with a fresh $ck=\bar A\us\Rq^{k'\times m}$ (the ring layer requires the ring keys and $pk^\ast$ to share one $A$, Section \ref{sec:pp}). It sets $pk_i:=pk^\ast$, runs $(pk_j,sk_j)\gets\KeyGen(pp)$ for all $j\ne i$, and gives $\mathcal A$ the keys $\{pk_j\}$ and $\{sk_j\}_{j\ne i}$. It answers, without ever using $sk_i$:
    
    \begin{itemize}
    \item $\Oracle{WSign}(\mu,\pi)$: draw $\rho\us S^m$, set $C=\Com_{ck}(\iota(\pi);\rho)$, $\sigma_r\gets\ASign(\inner{\mu,C},sk_{j_0(\pi)},\pi)$ using the known witness $sk_{j_0(\pi)}$, $\sigma=(\sigma_r,C,\pi\setminus\{pk_i\})$; store $(\mu,\pi,\sigma,\rho)$ in $W$; return $\sigma$.
    \item $\Oracle{Confirm}(\mu,\pi,\sigma)$: if $(\mu,\pi,\sigma,\rho)\in W$, query $\delta_s\gets\Oracle{Sign}(\inner{\mu,\pi,\sigma})$, set $M\gets M\cup\{\mu\}$, and return $\delta=(\delta_s,\iota(\pi),\rho)$; else $\bot$.
    \item Random oracles: $H_{\DS}$ is relayed to the $\DS$ challenger’s oracle; $H_{\Ambi}$ is simulated by $\mathcal B$ by lazy sampling ($\mathcal B$ ring-signs with the real witness $sk_{j_0(\pi)}$, so no programming is needed). Neither oracle is ever programmed by $\mathcal B$.
    \end{itemize}
\end{itemize}

This is a perfect emulation of $H_1$: the $\Oracle{WSign}$ answers are exactly those of $H_1$, and $\Oracle{Confirm}$ produces the honest $\delta$ because the stored $\rho$ opens $C$ to $\iota(\pi)$ and $\delta_s=\Sign_{sk_i}(\inner{\mu,\pi,\sigma})$ is what $\Oracle{Sign}$ returns (recall $sk_i=sk^\ast$).

Suppose $\mathcal A$ wins $H_1$, outputting $(\mu^\ast,\pi^\ast,\sigma^\ast,\delta^\ast)$ with $\delta^\ast=(\delta^\ast_s,\iota^\ast,\rho^\ast)$ for some index $\iota^\ast$, $\mu^\ast\notin M$, and $\CVerify=1$, i.e.\ $\Verify_{pk^\ast}(\inner{\mu^\ast,\pi^\ast,\sigma^\ast},\delta^\ast_s)=1$. The only messages $\mathcal B$ ever sent to $\Oracle{Sign}$ are the strings $\inner{\mu_t,\pi_t,\sigma_t}$ arising from confirmation queries, and every such $\mu_t$ was placed in $M$. Since $\mu^\ast\notin M$, we have $\mu^\ast\ne\mu_t$ for all $t$; by injectivity of $\inner{\cdot}$ the string $\inner{\mu^\ast,\pi^\ast,\sigma^\ast}$ differs from every queried string, hence was never signed. Therefore $(\inner{\mu^\ast,\pi^\ast,\sigma^\ast},\delta^\ast_s)$ is a valid $\DS$ forgery, which $\mathcal B$ outputs. Thus $\Pr[\mathcal A\text{ wins }H_1]\le\AdvEUF{\DS}(\mathcal B)$, with $\mathcal B$ making at most $q_C$ signing queries and never rewinding $\mathcal A$. Combining Steps 1--2 gives the theorem. (If the target index is chosen by $\mathcal A$ among the $N$ keys rather than fixed, $\mathcal B$ guesses it in advance, losing a factor $N$.)
\end{proof}

Claimability soundness makes the confirmation unambiguous: a confirmed signature must bind to a single signer, so that an adversary cannot present one withdrawable signature $\sigma$ with two valid confirmations attributing it to different ring members. This is the property that gives the committed index its meaning, without it, ``confirming'' would not pin down who signed, and it reduces directly to the binding of the commitment, hence to MSIS.

\begin{proposition}[Claimability soundness]\label{prop:claim}
For every PPT $\mathcal A$ there is a PPT $\mathcal A'$, running $\mathcal A$ once plus a parse, with $\Advclaim_{\EWS}(\mathcal A)\le\Advbind_{\Com}(\mathcal A')\le\AdvMSIS(\mathcal A')$.
\end{proposition}

\begin{proof}
Suppose $\mathcal A$ outputs $\sigma=(\sigma_r,C,aux)$ together with two confirmations $\delta=(\delta_s,i,\rho)$ and $\delta'=(\delta'_s,i',\rho')$, both accepted by $\Chk$ and with $i\ne i'$. Acceptance forces $C=\Com_{ck}(i;\rho)$ and $C=\Com_{ck}(i';\rho')$. Recalling $\Com_{ck}(x;\rho)=\bar A\rho+\Encode(x)$, subtracting the two openings gives
\[
\bar A\rho+\Encode(i)=\bar A\rho'+\Encode(i')\ \Longrightarrow\ \bar A(\rho-\rho')=\Encode(i')-\Encode(i).
\]
Put $\mathbf y=\bigl(\Encode(i)-\Encode(i'),\,\rho-\rho'\bigr)\in\Rq^{k'+m}$. Then
\[
[\,\mathbf I\mid\bar A\,]\cdot\mathbf y=\bigl(\Encode(i)-\Encode(i')\bigr)+\bar A(\rho-\rho')=\mathbf 0,
\] and $\mathbf y\ne\mathbf 0$ because $\Encode$ is injective and $i\ne i'$, while $\mathbf y$ is short
because $\Encode(i),\Encode(i')$ are fixed coset representatives and $\rho,\rho'\in S^m$. Hence $\mathbf y$ is a short nonzero kernel vector of $[\,\mathbf I\mid\bar A\,]$, i.e.\ an $\mathsf{MSIS}_{k',m,\gamma}$ solution (in the indexing of Definition \ref{def:msis}: the matrix is $\bar A\in\Rq^{k'\times m}$, the solution lives in $\Rq^{k'+m}$) with $\gamma=\lVert\mathbf y\rVert_\infty$. Every such $\mathcal A$ therefore yields a commitment-binding (equivalently MSIS) solver, giving, for the wrapper $\mathcal A'$ that runs $\mathcal A$ once and outputs the two openings (resp.\ the vector $\mathbf y$), $\Advclaim_{\EWS}(\mathcal A)\le\Advbind_{\Com}(\mathcal A')\le\AdvMSIS(\mathcal A')$.
\end{proof}

\begin{remark}[QROM]\label{rem:qrom}
The reductions use $H$ only through relaying (Theorem \ref{thm:unforge}) or not at all (Theorems \ref{thm:correct}--\ref{thm:wd}, Proposition \ref{prop:claim}). The only random-oracle programming is internal to $\Ambi$-anonymity and $\DS$-unforgeability; taking both in the QROM (the former via the honest-verifier zero-knowledge of the one-out-of-many proof, the latter via the Kiltz--Lyubashevsky--Schaffner analysis of Dilithium \cite{kiltz2018}, and the adaptive reprogramming of \cite{grilo2021}) makes all statements QROM statements. No $\EWS$-level reduction (Theorems \ref{thm:correct}--\ref{thm:unforge}, Proposition \ref{prop:claim}) programs $H$, each either relays it or ignores it, so all reprogramming is confined to the assumed QROM security of $\DS$ and $\Ambi$, and the $\EWS$-level reductions transfer verbatim with the standard adaptive-reprogramming adjustments. One step is not verbatim: the oracle each reduction \emph{simulates} rather than relays cannot be lazily sampled against superposition queries; in the QROM it is implemented by a $2q$-wise independent function, perfectly indistinguishable from a random oracle to any $q$-query quantum algorithm \cite{zhandry2021}.
\end{remark}

\section{Instantiation}\label{sec:instantiation}

Take $\DS$ to be a no-hint, full-$t$ Dilithium-style signature: challenge set $\Btau$ with $\tau\in\{39,49,60\}$, $\gamma_2=(q-1)/88$ (level 2) or $(q-1)/32$ (levels 3,5), not $\gamma_1/2$, and rejection $\lVert\mathbf z\rVert_\infty<\gamma_1-\beta$, $\lVert\LowBits(A\mathbf z-c\mathbf t,2\gamma_2)\rVert_\infty<\gamma_2-\beta$; its EUF-CMA security follows from SelfTargetMSIS and MSIS. Take $\Ambi$ to be a lattice one-out-of-many proof \cite{groth2015,esgin2019} over the same $A$, whose anonymity under full key exposure is its zero-knowledge; and $\Com$ the Ajtai commitment of Section \ref{sec:commit}.

\begin{corollary}\label{cor:instantiation}
With these instantiations, for every PPT $\mathcal A$,
\[
\Advwd_{\EWS}(\mathcal A)\le\AdvMLWE+\Advanon_{\Ambi},\qquad \Advclaim_{\EWS}(\mathcal A)\le\AdvMSIS,
\]
\[
\AdvEUF{\EWS}(\mathcal A)\le\AdvEUF{\DS}+q_W\,\Advanon_{\Ambi},
\]
all negligible under decisional MLWE, MSIS, SelfTargetMSIS and the zero-knowledge of $\Ambi$.
\end{corollary}

HAETAE may replace Dilithium as the base $\DS$ once its bimodal-Gaussian rejection is shown to meet the same EUF-CMA interface; it plays no role in the ambiguity or commitment layers. 

Quantitatively, the binding solution extracted in Proposition \ref{prop:claim} has norm $\gamma=\lVert\mathbf y\rVert_\infty\le\max(2B_{\Encode},\,2\lVert S\rVert_\infty)$, with $B_{\Encode}$ bounding the encoded indices (Remark \ref{rem:ajtai}) and $\lVert S\rVert_\infty$ the commitment randomness; the claimability bound is meaningful only when this $\gamma$ lies below the $\mathsf{MSIS}_{k',m,\gamma}$ hardness threshold (indexing as in Definition \ref{def:msis}) for the chosen $(q,n,k',m)$. A concrete parameter set simultaneously meeting $\DS$ EUF-CMA, $\Ambi$ zero-knowledge, MLWE-hiding and MSIS-binding at this $\gamma$ is left to the full version.

\begin{remark}[Why not a naive Dilithium challenge-split]\label{rem:naive}
The reason we do not instantiate $\Ambi$ by a naive Dilithium challenge-split, signing a ring with a single response $\mathbf z$ and challenges $c_1,\dots,c_N\in\Btau$ summing to $H(\cdot)$, recombining $c_i=H(\cdot)-\sum_{j\ne i}c_j$ for the signer, is given in Section \ref{sec:whyfails}: Proposition \ref{prop:notranscription} shows the honest signer accepts only with probability $2^{-\Omega(n)}$, and Remark \ref{rem:anonfail} shows anonymity would fail even if completeness were patched. (Should one still attempt it, the correct low-order acceptance test is on $\LowBits(A\mathbf z-\sum_j c_j\mathbf t_j,2\gamma_2)$, i.e.\ on the recombined verification vector $\mathbf w=A\mathbf z-\sum_j c_j\mathbf t_j$, not on $A\mathbf z-c_i\mathbf t_i$.)
\end{remark}

\section{Comparison with Liu-Susilo-Baek}\label{sec:comparison}
This section collects, in one place, every comparison with the discrete-logarithm construction of Liu, Susilo and Baek \cite{liu2024}: the structural reason a literal port fails (Section \ref{sec:compstruct}), the definition- and theorem-level correspondence with our security model (Section \ref{sec:compmodel}, Table \ref{tab:comparison}), and the algorithm-by-algorithm differences of the construction (Section \ref{sec:compalg}). The definitional refinements themselves are stated where the definitions are introduced (Section \ref{sec:relation}); here we assemble the full picture and justify each divergence against a single structural fact, the absence, over lattices, of a free re-randomisation of a shared value.

\subsection{The structural obstruction}\label{sec:compstruct}
Liu, Susilo and Baek realise extended withdrawability from a discrete-logarithm one-out-of-$N$ signature in which re-randomising the signer's contribution across the ring is essentially free: one group element serves ambiguity and confirmation. The lattice setting has no such free re-randomisation. Transporting a signer-derived shift in the clear (as an earlier draft did, publishing $\{A\mathbf s_1 s_j\}_j$) exposes $A\mathbf s_1\approx\mathbf t_{\mathsf{signer}}$, since $\mathbf t_i-A\mathbf s_1=\mathbf s_2$ is short only for the true signer, and breaks ambiguity, the same structural obstruction met in the designated-verifier setting. We therefore separate the two roles: ambiguity is a genuine lattice one-out-of-$N$ proof, and confirmation is a claim, a binding signature plus the opening of a hiding index commitment bound into the transcript. This is exactly a claimable ring signature \cite{park2019,yamashita2023}, a primitive the related-work literature already identifies with withdrawable signatures. The cost is one commitment and its opening per signature; the gain is that no signer-identifying value is ever published, so ambiguity reduces cleanly to MLWE (hiding) plus the zero-knowledge of the one-out-of-$N$ proof, and confirmation soundness to MSIS-binding plus the base scheme's EUF-CMA.

\subsection{Correspondence with the Liu-Susilo-Baek security model}\label{sec:compmodel}
Table \ref{tab:comparison} sets each notion of our model against its counterpart in \cite{liu2024}; the per-notion justifications are given in Section \ref{sec:relation}, where each definition is introduced. The two rows that carry genuine security content, rather than a change of algebraic setting, are extended withdrawability (a hiding game in place of a completeness statement) and claimability soundness (a new notion with no counterpart in \cite{liu2024}); the unforgeability row records a strictly tighter reduction (no multiplicative $q_H$ loss).

\begin{table}[t]
\centering
\small
\renewcommand{\arraystretch}{1.25}
\begin{tabularx}{\textwidth}{@{}l >{\raggedright\arraybackslash}p{4.4cm} >{\raggedright\arraybackslash}X@{}}
\toprule
\textbf{Notion} & \textbf{Liu-Susilo-Baek \cite{liu2024}} & \textbf{This work (difference and justification)}\\
\midrule
Syntax & seven algorithms & Definition \ref{def:syntax}; identical, but $\Confirm$, $\CVerify$ also read $\sigma,\pi$, as the confirmed signature binds the whole object $\inner{\mu,\pi,\sigma}$\\ Correctness & Definition 1 & Definition \ref{def:correctness}; identical\\ Unforgeability & Definition 2; claimed loss $L=q_H$ & Definition \ref{def:eufcmaEWS}; fixed target (loss $N$), no $q_H$ factor (Thm. \ref{thm:unforge})\\ Extended withdrawability & Definition 3: completeness, $\Pr=1$ & Definition \ref{def:wd}: hiding game, $\lvert\Pr-\tfrac12\rvert$ negligible; completeness is not a hiding guarantee\\ Signer ambiguity & Definition 4, $\le\tfrac1N+\negl$ & Definition \ref{def:anon} (full key exposure); same notion, moved to the building block and folded into withdrawability\\ Claimability soundness &, (only implicit in $\Chk$) & Definition \ref{def:claim}, Prop. \ref{prop:claim}; new, makes ``claim'' sound via commitment binding\\
Assumptions & discrete logarithm & MLWE, MSIS, SelfTargetMSIS $+$ commitment; post-quantum, and the commitment replaces the free $g^r$\\ Unforgeability proof & Thm. 1: loss $q_H$, informal extraction & Thm. \ref{thm:unforge}: no $q_H$ factor, explicit reduction\\ 
Ambiguity proof & Thm. 2: programming step left implicit & building-block ZK (Definition \ref{def:anon}); the implicit step is fatal over lattices (Prop. \ref{prop:notranscription})\\ Withdrawability proof & Thm. 3: informal $1/N$ $+$ completeness & Thm. \ref{thm:wd}: reduction to hiding $+$ anonymity\\
Impossibility result &, & Prop. \ref{prop:notranscription}: no literal lattice transcription\\
\bottomrule
\end{tabularx}
\caption{Our security model against that of Liu, Susilo and Baek \cite{liu2024}. Construction-level differences are compared in Section \ref{sec:comparison}.}
\label{tab:comparison}
\end{table}

\subsection{Construction, algorithm by algorithm}\label{sec:compalg}
We now compare the seven algorithms of Section \ref{sec:construction} with their discrete-logarithm originals. The first three items contrast the signing-side algorithms directly; the last collects the same three contrasts against a \emph{literal} lattice transcription of \cite{liu2024}, making explicit which defect each of our design choices repairs.

\paragraph{Key generation vs.\ their $g^{sk}$.}
Their key generation draws $sk\us\mathbb Z_p$ and sets $pk=g^{sk}$; ours is the module-lattice analogue, with $sk=(\mathbf s_1,\mathbf s_2)$ short and $pk=A\mathbf s_1+\mathbf s_2$. They require no $\Setup$ because the group $G$ and generator $g$ are global objects; the lattice setting has no canonical counterpart, so the common $A$ must be produced by $\Setup$ and shared, precisely the amendment missing from a literal transcription (Section \ref{sec:comparison}).

\paragraph{Withdrawable signing vs.\ their published shifts.}
Their $\WSign$ re-randomises every ring key by a shift, forming $\gamma'=\{pk_j g^{s_j}\}_j$, signs $\gamma'$ with a shifted witness $\widetilde{sk}=sk_i+r$, and publishes the shifts $\sigma_2=\{g^{s_j}\}_j$ so the verifier can rebuild $\gamma'$; the shift exists solely to enable the linked confirmation. We keep the ring $\pi$ and the witness $sk_i$ unchanged and introduce a committed index instead. The reason is structural: over a prime-order group the published shifts $\{g^{s_j}\}$ are uniform and independent of the secret, hence harmless, whereas their literal lattice analogue $\{A\mathbf s_1 s_j\}$ shares the common factor $A\mathbf s_1\approx\mathbf t_i$ and thereby reveals the signer (Section \ref{sec:comparison}). Replacing the shift by a hiding commitment removes the leak while preserving the one function the shift served, carrying, into $\sigma$, a signer-chosen value that confirmation can later open.

\paragraph{Confirmation vs.\ their ephemeral-key claim.}
Their confirmation re-derives the shared randomness $r=H(\gamma,g^{sk_i\sum_j s_j})$ and outputs $\delta_1=\Sign(\{m,\sigma\},r)$ (a signature under the ephemeral key $g^r$), $\delta_2=\Sign(m,sk_i)$ (under the real key), and $\delta_3=g^r$; their $\Chk$ then verifies $g^r=g^{s_i}$ against the value embedded in $\sigma_2$ and $\delta_1$ under $g^r$. Our single signature $\delta_s$ on $\inner{\mu,\pi,\sigma}$ plays the role of both their $\delta_2$ (it binds the real key $pk_i$) and their $\delta_1$ (it binds the object $\sigma$), while the commitment $C$ plays the role of the embedded $g^r$: the signer-chosen value fixed inside $\sigma$ at signing time that $\Chk$ re-derives and matches at confirmation. The commitment is the lattice analogue of $g^r$, with one strict improvement, it hides the index on its own, whereas $g^{s_i}$ hides only in combination with the ambiguity of $\sigma_1$. This is why our $\Confirm$ needs one base signature plus one opening, in place of two signatures and an embedded group element.

\subsection{Algorithm-level relation to a direct lattice transcription}\label{sec:algrelation}
The construction above is most sharply understood as the correction of a direct lattice transcription of \cite{liu2024}, which realises the seven algorithms over $\Rq$ but ports the discrete-logarithm operations symbol for symbol. Each of that transcription's signing-side algorithms is thereby defective, and our algorithms differ in exactly the places that repair them.

\begin{itemize}
    \item Key generation: The transcription's $\KeyGen$ samples a fresh $A\us\Rq^{k\times l}$ inside each invocation, so distinct users hold distinct matrices and the ring identity $A\mathbf z-\sum_j c_j\mathbf t_j$ is not defined across $\pi$. Our $\Setup$ (Algorithm \ref{alg:setup}) fixes one common $A\in\Rq^{k\times l}$ shared by all $N$ users, making the ring statement well formed, the omission a group-based transcription cannot see, since $g\in G$ is global by fiat.
    \item Withdrawable signing: The transcription's $\WSign$ shifts each ring key by $A\mathbf s_1 s_j$, publishes the shifts $\sigma_2=\{A\mathbf s_1 s_j\}_j$ so the verifier can rebuild the shifted ring, and signs that ring with the shifted witness $\widetilde{sk}=((1+r)\mathbf s_1,\mathbf s_2)$, where $r=s_i\in\Btau$. Both steps are unsound over lattices. First, the published shifts share the common factor $A\mathbf s_1=\mathbf t_i-\mathbf s_2\approx\mathbf t_i$, which is short-distance from the signer's key $\mathbf t_i$ but from no other $\mathbf t_j$; publishing them therefore identifies the signer and destroys ambiguity. Second, $(1+r)\mathbf s_1$ is no longer short: for $r\in\Btau$ one has $\lVert(1+r)\mathbf s_1\rVert_\infty\le(1+\tau)\lVert\mathbf s_1\rVert_\infty$, so the shifted witness violates the norm bound that $\ASign$'s rejection sampling presumes (and on which its zero-knowledge rests). Our $\WSign$ (Algorithm \ref{alg:wsign}) publishes no shift, leaves the witness $sk_i$ and the ring $\pi$ unchanged, and carries the signer-chosen value in the hiding commitment $C$ instead.
    \item Confirmation: The transcription's $\Confirm$ sets $\aleph=\sigma_2\setminus\{A\mathbf s_1 s_i\}$ and re-derives $r=H\bigl(\pi,\HighBits(\sum_{j\ne i}A\mathbf s_1 s_j+\mathbf s_2,2\gamma_2)\bigr)$; but recovering $A\mathbf s_1 s_i$ presupposes $s_i=r$, whose derivation needs $\aleph$, a circular dependency, well defined only if $\sigma_2$ is index-ordered or the per-signature randomness is stored. It then uses $\delta_3=A(1+r)\mathbf s_1+\mathbf s_2$ as an on-the-fly public key against which one of the confirmed components is verified. Our $\Confirm$ (Algorithm \ref{alg:confirm}) is stateless and non-circular: the signer signs $\inner{\mu,\pi,\sigma}$ under its real key $sk_i$ and opens the commitment it stored at signing, with no shift to recover and no ephemeral key to synthesise.
\end{itemize}

In one sentence: the transcription fails because it treats the group re-randomisation $g^{s_j}$ as though it had a free lattice analogue; the present construction removes the re-randomisation altogether and recovers traceability from a committed index. The single structural cause behind all three defects, the absence of a cheap lattice shared secret, is analysed in Section \ref{sec:comparison}.

\section{Conclusion}\label{sec:conclusion}

We have given a lattice realisation of extended withdrawable signatures over the Fiat-Shamir with aborts paradigm. The central move is a change of viewpoint: rather than porting the discrete-logarithm construction of Liu, Susilo and Baek \cite{liu2024} operation by operation, which, as we prove, cannot be done securely (Proposition \ref{prop:notranscription}), we recognise the extended withdrawable signature as a \emph{claimable ring signature} in the sense of Park and Sealfon \cite{park2019}, the very abstraction that \cite{liu2024} single out as an instantiation of their primitive. This reframing separates the two responsibilities of the object and assigns each to a component with an established analysis: signer ambiguity is carried by an anonymous one-out-of-$N$ proof used as a black box, and confirmation is a signer's claim, a binding base signature on the whole withdrawable object together with the opening of a hiding index commitment bound into the transcript. No signer-derived value is ever published in the clear, which is precisely what a literal transcription fails to achieve.

\paragraph{What we proved.} Within a security model that sharpens \cite{liu2024} in two respects, extended withdrawability is recast from a completeness statement into an anonymity-until-claim indistinguishability game (Definition \ref{def:wd}), and claimability soundness is added as an explicit notion (Definition \ref{def:claim}), we gave complete, fully explicit game-based reductions for all four properties in the (quantum) random-oracle model. Correctness (Theorem \ref{thm:correct}) follows modularly from the correctness of the ingredients. Extended withdrawability (Theorem \ref{thm:wd}) reduces, through a three-game sequence, to commitment hiding (hence decisional MLWE) and the anonymity of the one-out-of-$N$ layer. Unforgeability under insider corruption (Theorem \ref{thm:unforge}) reduces to the plain EUF-CMA security of the base signature, with a clean forgery extraction and, unlike \cite[Thm. 1]{liu2024}, no rewinding and no multiplicative $q_H$ loss, the only overheads being the factor $N$ from guessing the target and an additive $q_W\Advanon_{\Ambi}$ term. Claimability soundness (Proposition \ref{prop:claim}) reduces directly to commitment binding (hence MSIS). Along the way, Proposition \ref{prop:notranscription} isolates a lattice-specific impossibility with no counterpart in the discrete-logarithm setting: the additive challenge-split that makes their one-out-of-$N$ signature work is complete only with negligible probability once challenges are confined to the sparse set $\Btau$, and anonymity fails even if completeness is forced.

\paragraph{Significance.} The construction is publicly verifiable, ambiguous until claimed, and post-quantum: its security rests entirely on module-lattice assumptions (MLWE, MSIS, SelfTargetMSIS) and the unforgeability of a standardised-style base signature, so it retains its guarantees against a quantum adversary, exactly the durability that long-lived commitments in decentralised systems require, and that the discrete-logarithm foundation of \cite{liu2024} lacks. Because ambiguity, binding and the base signature are invoked only through their stated interfaces, the result is also modular: any anonymous one-out-of-$N$ proof, any hiding-and-binding commitment, and any EUF-CMA base signature meeting the interfaces of Section \ref{sec:construction} may be substituted without reopening the proofs. Instantiated with a no-hint, full-$t$ Dilithium-style base and an established logarithmic-size lattice one-out-of-many proof (Section \ref{sec:instantiation}), the scheme is concretely realisable, and the design is faithful to the primitive, the security goals, and the stated intuition of \cite{liu2024}.

\paragraph{Future work.} Several directions remain. First, a fully concrete instantiation: pinning down a logarithmic-size one-out-of-$N$ proof \cite{groth2015,esgin2019} with explicit parameters and giving a single parameter set that simultaneously meets base-signature EUF-CMA, one-out-of-$N$ zero-knowledge, MLWE-hiding and MSIS-binding at the extracted norm $\gamma$ of Proposition \ref{prop:claim}, the quantitative gap flagged in Section \ref{sec:instantiation}. Second, a rigorous HAETAE \cite{cheon2024} base: verifying that its bimodal-Gaussian rejection meets the same EUF-CMA interface would yield shorter signatures at no cost to the surrounding analysis. Third, tightening the reductions: both the factor $N$ from target-guessing in Theorem \ref{thm:unforge} and the factor $2$ relating the guessing and distinguishing forms of the anonymity advantage (Definition \ref{def:anon}) are artefacts of the current arguments, and it is natural to ask whether either can be removed. Fourth, functionality extensions that the modular structure should accommodate cleanly: blind withdrawable signing, threshold or multi-signer confirmation, and revocation or time-bounded confirmation windows. Finally, a treatment of concurrent confirmation and an analysis in the plain model (without random oracles) for the ambiguity layer would further strengthen the guarantees.

\bibliographystyle{plain}
\bibliography{Bibliography.bib}

\end{document}